\documentclass{article}
\usepackage[english]{babel}
\usepackage[utf8]{inputenc}
\usepackage{times}
\usepackage{soul}
\usepackage{amsmath}
\usepackage{amssymb}
\usepackage{amsthm}
\usepackage{url}
\usepackage{booktabs}
\usepackage{fontawesome}
\usepackage{utfsym}
\usepackage{wasysym}
\usepackage{makecell}
\usepackage{algorithm}
\usepackage{algpseudocode}
\usepackage{pifont}
\newcommand{\cmark}{\ding{51}} 
\newcommand{\xmark}{\ding{55}} 
\newcommand{\qmark}{?}

\newtheorem{theorem}{Theorem}

\newtheorem{lemma}{Lemma}
\newtheorem{definition}{Definition}
\newtheorem{observation}{Observation}
\newtheorem{proposition}{Proposition}
\usepackage[a4paper,top=2cm,bottom=2cm,left=3cm,right=3cm,marginparwidth=1.75cm]{geometry}
\usepackage{graphicx}
\usepackage[colorlinks=true, allcolors=blue]{hyperref}
\usepackage{comment}
\usepackage[small]{caption}
\usepackage{booktabs}
\usepackage[switch]{lineno}
\usepackage{xcolor}
\usepackage{dblfloatfix}
\usepackage{authblk}
\usepackage{natbib}

\title{Temporal Fair Division of Indivisible Goods with Scheduling}
\author[1]{Kui-Wang Choi}
\author[2]{Minming Li}
\affil[1,2]{Department of Computer Science, City University of Hong Kong}
\affil[ ]{\textit {kuiwchoi2-c@my.cityu.edu.hk, minming.li@cityu.edu.hk}}

\date{}

\begin{document}
\maketitle

\begin{abstract}
We study temporal fair division, where agents receive goods over multiple rounds and cumulative fairness is required. We investigate Temporal Envy-Freeness Up to One Good (TEF1) and Up to Any Good (TEFX), its approximation $\alpha$-TEFX, and Temporal Maximin Share (TMMS). Motivated by known impossibilities in standard settings, we consider the model in various restricted settings and extend it by introducing scheduling.

Our main contributions draw the boundary between possibility and impossibility. First, regarding temporal fair division without scheduling, we prove that while constant-factor $\alpha$-TEFX is impossible in general, a $1/2$-approximation is achievable for generalized binary valuations and identical days with two agents. Second, regarding temporal fair division with scheduling, we demonstrate that a scheduling buffer of size at least $n/2$ enables TEF1 for identical days. However, we establish that TEFX and TMMS remain largely impossible even with scheduling or restricted domains. These results highlight the inherent difficulty of strict temporal fairness and quantify the trade-offs required to achieve approximation guarantees.
\end{abstract}

\section{Introduction}
Fair division is a fundamental problem in multi-agent systems and economics, traditionally focusing on the static allocation of resources, such as splitting a cake or distributing a set of goods at a single point in time \citep{lipton2004,budish2011,gourves2014,caragiannis2019,plaut2017,aziz2021,chan2019,amanatidis2020}. However, many real-world allocation problems are dynamic. From cloud computing resources and food bank donations to organ exchange networks, goods often arrive over time and must be allocated online \citep{aleksandrov2015,ghodsi2013,bertsimas2011,lee2019,benade2024}. Although this approach captures the dynamic feature for real-world problems, it ignores the fairness in intermediate allocations. Therefore, the goal of temporal fair division is to maintain fairness cumulatively over the long run, rather than just instantaneously in each round.

The transition from dynamic to temporal settings introduces significant theoretical challenges. Unlike dynamic fair division, where intermediate allocations need not be fair, temporal mechanisms must account for this. Recent work has painted a bleak picture for this domain: temporal fair allocations are often impossible to guarantee without compromising efficiency or social welfare \citep{he2019,elkind2024,cookson2024}. For instance, strong fairness concepts like Envy-Freeness up to any good (EFX) are frequently unattainable in temporal settings because an agent's perception of fairness fluctuates wildly as new goods arrive.

These impossibility results from the temporal fair division model raise a natural question: Are these failures inherent to the problem, or are they artifacts of the model's rigidity? Standard temporal fair division assumes that goods must be consumed or allocated immediately upon arrival. However, in many applications, such as job scheduling on a server or assigning shifts, a natural flexibility exists: A system administrator can delay a job slightly; a manager can reshuffle next week's roster; a student affairs office can assign part-time jobs to students within a flexible interval \citep{schwiegelshohn2000,li2021,im2020}.

Therefore, our main research question is to \textit{examine the boundaries of temporal fairness across various settings, with or without scheduling, in temporal fair division}.

\subsection{Main Results}
We list our detailed results and previous results in the temporal fair division without scheduling in Table~\ref{tab:tef}, which is considered in Section~\ref{woschedule}, and our results with scheduling in Table~\ref{tab:tefR}, which is considered in Section~\ref{wschedule}.
\begin{itemize}
    \item \textbf{Allocations without Scheduling:} We establish that TEF1 allocations exist under the house allocation settings when \(T = 3\) and can be computed in polynomial time, while $\alpha$-TEFX allocations have guaranteed lower bounds in most settings. However, under various settings, we demonstrate that exact TEFX yields mixed results and that TMMS allocations may not exist.
    \item \textbf{Allocations with Scheduling:} We show that scheduling improves fairness achievability in specific cases—guaranteeing TEF1 with a buffer of at least $n/2$ rounds with identical days. But we prove that scheduling alone is insufficient to guarantee the general existence of either TEFX or TMMS allocations.
\end{itemize}

\begin{table*}[htbp]
    \small
  \centering
  \caption{Possibilities, impossibilities, and open questions in temporal fair division. \cmark indicates a possibility result; \xmark indicates an impossibility; \qmark indicates an open question}
  \label{tab:tef}
    \begin{tabular}{|c|c|c|c|c|}
      \hline
       & TEF1 & TEFX & \(\alpha\)-TEFX & TMMS \\
      \hline
      General Settings & \makecell{\xmark for \(n \geq 3\) \\ \cmark for \(n = 2\) \\ \citep{he2019} \\ \cmark for \(T = 2\) \\ \citep{elkind2024}} & \makecell{\xmark \\ \citep{elkind2024}} & \makecell{\cmark \(\frac{\frac{1}{2} \cdot min_{g \in O} v_i(g)}{min_{g \in O} v_i(g) + \frac{1}{2} \cdot max_{g \in O} v_i(g)}\) \\ (\(v_i(g) > 0\)) \\(\textcolor{red}{Theorem~\ref{thm:atefx}})} & \makecell{\xmark \\ (\textcolor{red}{Theorem~\ref{thm:tmmsgeneralsettingR}})} \\
      \hline
      Identical Days & \makecell{\cmark for \(T = 3\) \\ House Allocation Settings \\ (\textcolor{red}{Theorem~\ref{thm:tef1identicaldaysHouse}})} & \makecell{\xmark for \(T \geq 3\) \\ (\textcolor{red}{Theorem~\ref{thm:tefxidenticaldays}}) \\ \qmark for \(T = 2\)} & \makecell{\cmark \(\frac{1}{2}\) for \(n = 2\) \\ (\textcolor{red}{Theorem~\ref{thm:atefxidenticaldays}})} & \makecell{\xmark \\ (\textcolor{red}{Theorem~\ref{thm:tmmsidenticaldays}})} \\
      \hline
      \makecell{Generalized\\Binary Valuation} & \makecell{\cmark \\ \citep{elkind2024}} & \makecell{\cmark for \(n = 2\) \\ (\textcolor{red}{Theorem~\ref{thm:tefxgenbivalue}})} & \makecell{\cmark \(\frac{1}{2}\) \\ (\textcolor{red}{Theorem~\ref{thm:atefxgenbivaluation}})} & \makecell{\cmark \\ \citep{elkind2024}} \\
      \hline
      Identical Valuation & \makecell{\cmark \\ \citep{elkind2024}} & \makecell{\xmark \\ \citep{elkind2024}} & \makecell{\cmark \(\frac{min_{g \in O} v(g)}{max_{g \in O} v(g) + min_{g \in O} v(g)}\) \\ (\(v(g) \in R_{\{0\} \cup [1, \infty)}\)) \\ (\textcolor{red}{Theorem~\ref{thm:atefxidenticalvaluation}})} & \makecell{\xmark \\ (\textcolor{red}{Theorem~\ref{thm:tmmsgeneralsettingR}})} \\
      \hline
      Bi-valued Goods & \makecell{\qmark} & \makecell{\xmark \\ \citep{elkind2024}} & \makecell{\cmark \(\frac{a}{b}\) \\ \(b \geq a\) \\ (\textcolor{red}{Theorem~\ref{thm:atefxbivalued}})} & \makecell{\xmark \\ (\textcolor{red}{Theorem~\ref{thm:tmmsgeneralsettingR}})} \\
      \hline
    \end{tabular}
\end{table*}
\begin{table*}[ht]
\small
  \centering
  \caption{Possibilities, impossibilities, and open questions in temporal fair division with scheduling. \cmark indicates a possibility result; \xmark indicates an impossibility; \qmark indicates an open question; \cmark* indicates this result is implied by the positive result without scheduling.}
  \label{tab:tefR}
    \begin{tabular}{|c|c|c|c|}
      \hline
       & TEF1 & TEFX & TMMS \\
      \hline
      General Settings & \makecell{\cmark* \\ \(r \geq \frac{T}{2}\) \\ \citep{elkind2024}} & \makecell{\xmark \\ (\textcolor{red}{Theorem~\ref{thm:tefxgeneralsettingsR}})} & \makecell{\xmark \\ (\textcolor{red}{Theorem~\ref{thm:tmmsgeneralsettingR}})} \\
      \hline
      Identical Days & \makecell{\cmark \\ \(r \geq \frac{n}{2}\) \\ (\textcolor{red}{Theorem~\ref{thm:tef1identicalDaysRn/2}})} & \makecell{\qmark} & \makecell{\qmark} \\
      \hline
      \makecell{Generalized\\Binary Valuation} & \makecell{\cmark* \\ \citep{elkind2024}} & \makecell{\qmark} & \makecell{\cmark* \\ \citep{elkind2024}} \\
      \hline
      Identical Valuation & \makecell{\cmark* \\ \citep{elkind2024}} & \makecell{\xmark \\ (\textcolor{red}{Theorem~\ref{thm:tefxgeneralsettingsR}})} & \makecell{\xmark \\ (\textcolor{red}{Theorem~\ref{thm:tmmsgeneralsettingR}})} \\
      \hline
      Bi-valued Goods & \makecell{\qmark} & \makecell{\xmark \\ (\textcolor{red}{Theorem~\ref{thm:tefxgeneralsettingsR}})} & \makecell{\xmark \\ (\textcolor{red}{Theorem~\ref{thm:tmmsgeneralsettingR}})} \\
      \hline
    \end{tabular}
\end{table*}
\subsection{Related Works}
Our work is related to online fair division \citep{aleksandrov2015,benade2024,aleksandrov2017,wang2026}, where some (possibly zero) items arrive in consecutive rounds. We study a stronger version of it, known as temporal fair division, in which fairness must hold at every round \citep{he2019,elkind2024,cookson2024}. \cite{he2019} showed that maintaining EF1 for goods is impossible in an uninformed setting without reallocating past items, and proposed a reallocation-based algorithm for two agents. However, they provided a polynomial-time algorithm, where at each round, the allocation is EF1 in the informed setting with two agents. \cite{elkind2024} formalized temporal fairness for goods and chores, proving that TEFX allocations may not exist and providing TEF1 algorithms for two agents with chores or restricted settings: two item types, generalized binary valuations, and unimodal preferences. Finally, \cite{cookson2024} explored SD-EF1, EF1, and PROP1, focusing on achieving both temporal and daily fairness. They demonstrated that, while per-round SD-EF1 is generally impossible, under the identical days setting, it is always attainable, guaranteeing a final SD-PROP1 allocation.

\section{Preliminaries}
First, \(\forall k \in \mathbb{N}\), denote \([k] := \{1, \ldots, k\}\).
We study temporal fair division under an informed online fair division setting, denoted by the tuple \(\mathcal{I} = \left( N, T, \left\{ O_t \right\}_{t \in [T]}, v = (v_1, \ldots, v_n), r \right)\).

We have a set of \(N = [n]\) agents, and a set of items \(\left\{ O_t \right\}\) that arrive in round \(t \in [T]\), which are to be allocated to the agents.
Denote \(O^{t} = \bigcup_{i=1}^{t} O_i\) as the cumulative set of goods up to round \(t\), and \(O = O^{T} = \bigcup_{i=1}^{T} O_i\) as the set of all goods.
For each agent, there is an additive valuation function \(v_i: 2^{O} \longrightarrow \mathbb{R}\) over the items, where for \(S \subseteq O\), \(v_i(S) = \sum_{g \in S} v_i(\{g\})\). We slightly abuse the notation and assume that \(v_i(g) = v_i(\{g\})\). In addition, this work focuses on goods allocation, i.e. \(\forall g \in O, \forall i \in N, v_i(g) \geq 0\).

An allocation \(A = (A_1, \dots, A_n)\) is defined as an ordered partial partition of \(O\), where \(A_i\) is the goods allocated to agent \(i\), such that \(\forall i, j \in [N]\) where \(i \neq j\), \(A_i \cap A_j = \emptyset\) and \(\bigcup_{i \in \mathcal{N}} A_i = O\).

In addition, under the online fair division setting, denote \(A^t = (A^t_1, \dots, A^t_n)\) as the allocation from round \(1\) to after round \(t\), that is \(A^t_i = A^{t-1}_i \cup (A_i \cap O^t)\).

In Section~\ref{woschedule}, we assume the scheduling buffer \(r = 1\). We defer the detailed definition of \(r\) to Section~\ref{wschedule}.

To formally evaluate allocations within our temporal framework, we build upon the standard fairness taxonomy established in the static fair division literature \citep{amanatidis2023}. The ideal benchmark is exact Envy-Freeness (EF) \citep{varian1974}, which dictates that no agent strictly prefers another agent's bundle to their own. However, when goods are indivisible, exact EF is notoriously impossible to guarantee in all instances.

This inherent limitation necessitates practical relaxations. The most common baseline is Envy-Freeness up to One Good (EF1) \citep{lipton2004,budish2011}, which permits envy provided it vanishes upon the hypothetical removal of the most valuable item from the envied agent's bundle. A strictly stronger and more desirable standard is Envy-Freeness up to Any Good (EFX) \citep{caragiannis2019,gourves2014}, requiring the envy to dissipate even if the least valuable item is removed. Because the existence of exact EFX remains a profound open problem for $n > 3$ agents, the literature frequently relies on multiplicative approximations, denoted as $\alpha$-EFX \citep{plaut2017}. Under $\alpha$-EFX, an agent’s valuation of their own bundle must be at least a fraction $\alpha \in (0, 1]$ of their valuation of the envied bundle after removing any single item. Alongside these envy-based metrics, we also evaluate Maximin Share (MMS) \citep{budish2011}, a threshold-based guarantee ensuring that each agent receives a bundle worth at least what they could secure if they acted as the divider and received the worst resulting share.

The transition from static to dynamic environments profoundly disrupts the feasibility of these guarantees. In the temporal setting, we denote the perpetual enforcement of these properties by prefixing them with "T" (i.e., TEF1, TEFX, $\alpha$-TEFX, TMMS), dictating that the respective fairness condition must be satisfied cumulatively at every discrete round $t \in [T]$ \citep{elkind2024}.

This temporal requirement fundamentally alters the landscape of existence bounds. As our subsequent analysis will demonstrate, the relentless need to maintain immediate, round-by-round equity causes exact TEFX and TMMS to frequently collapse into impossibility across multiple valuation domains. Consequently, establishing robust $\alpha$-TEFX approximations serves as our primary mathematical vehicle for guaranteeing structured equity in these highly constrained temporal environments.

\begin{definition} [Envy-Freeness (EF)] 
An allocation \(A = (A_1, \dots, A_n)\) is envy-free (EF) if \(\forall i, j \in N, v_i(A_i) \geq v_i(A_j)\).
\end{definition}

\begin{definition} [Temporal Envy-Freeness (TEF)] 
For any \(t \in [T]\), an allocation \(A^t = (A^t_1, \dots, A^t_n)\) is temporal envy-free (TEF) if \(\forall t' \leq t, A^{t'}\) is EF.
\end{definition}

\begin{definition} [Envy-Freeness Up to One Good (EF1)] 
An allocation \(A = (A_1, \dots, A_n)\) is envy-free up to one good (EF1) if \(\forall i, j \in N, \exists g \in A_j\) s.t. \(v_i(A_i) \geq v_i(A_j \backslash \{g\})\).
 \end{definition}

\begin{definition} [Temporal Envy-Freeness Up to One Good (TEF1)]
For any \(t \in [T]\), an allocation \(A^t = (A^t_1, \dots, A^t_n)\) is temporal envy-free up to one good (TEF1) if \(\forall t' \leq t, A^{t'}\) is EF1.
\end{definition}

\begin{definition} [Envy-Freeness Up to Any Good (EFX)]
An allocation \(A = (A_1, \dots, A_n)\) is envy-free up to any good (EFX) if \(\forall i, j \in N, \forall g \in A_j\) s.t. \(v_i(A_i) \geq v_i(A_j \backslash \{g\})\).
\end{definition}

\begin{definition} [Temporal Envy-Freeness Up to Any Good (TEFX)]
For any \(t \in [T]\), an allocation \(A^t = (A^t_1, \dots, A^t_n)\) is temporal envy-free up to any good (TEFX) if \(\forall t' \leq t, A^{t'}\) is EFX.
\end{definition}

\begin{definition} [Approximately Envy-Freeness Up to Any Good (\(\alpha\)-EFX)]
Let \(\alpha \in (0, 1]\). An allocation \(A = (A_1, \dots, A_n)\) is approximately envy-free up to any good (\(\alpha\)-EFX) if \(\forall i, j \in N, \forall g \in A_j\) s.t. \(v_i(A_i) \geq \alpha \cdot v_i(A_j \backslash \{g\})\).
\end{definition}

\begin{definition} [Approximately Temporal Envy-Freeness Up to Any Good (\(\alpha\)-TEFX)]
For any \(t \in [T]\), an allocation \(A^t = (A^t_1, \dots, A^t_n)\) is approximately temporal envy-free up to any good (\(\alpha\)-TEFX) if \(\forall t' \leq t, A^{t'}\) is \(\alpha\)-EFX.
\end{definition}

\begin{definition} [Maximin Share Fairness (MMS)]
Denote \(M_n(O)\) to be the collection of all possible allocations of the goods in \(O\) to \(n\) agents. An allocation \(A = (A_1, \dots, A_n)\) is maximin share fair (MMS) if \(\forall i \in N\), \(v_i(A_i) \geq \mu_i^n(O) = \max_{B \in M_n(O)} \min_{S \in B} v_i(S)\).
\end{definition}

\begin{definition}[Temporal Maximin Share Fairness (TMMS)]
For any \(t \in [T]\), an allocation \(A^t = (A^t_1, \dots, A^t_n)\) is temporal maximin share fair (TMMS) if \(\forall t' \leq t, A^{t'}\) is MMS.
\end{definition}

Having established the foundational metrics of temporal fairness, we now define the specific operational constraints under which we will evaluate these properties. To systematically examine the boundaries of temporal fair division, we analyze several structured settings, beginning with the Identical Days environment \citep{igarashi2024b}.

Intuitively, this setting dictates that the exact same set of items—with identical valuation profiles—arrives at every discrete time step. While mathematically rigid, this structure directly models pervasive real-world allocation paradigms characterized by strict periodicity. For example, in cloud computing and network bandwidth distribution, a static pool of server resources or time-slots is continuously provisioned and allocated in recurring, identical cycles \citep{baruah1995}.
\begin{definition}[Identical Days]
    A temporal fair division instance is under the identical days setting iff \(\forall t, t' \in [T]\), there is a bijection \(f_{t, t'}: O_t \Rightarrow O_{t'}\) s.t. \(\forall i \in [n], \forall g \in O_t, v_i(g) = v_i(f_{t, t'}(g))\).
\end{definition}

The second environment we investigate is the Generalized Binary Valuation setting \citep{halpern2020}. Under this constraint, agents exhibit strictly dichotomous preferences: an agent values an arriving good at either zero or a fixed positive utility $b$. This mathematical restriction effectively models practical allocation scenarios characterized by strict, needs-based utility. In many real-world contexts, individuals possess highly polarized preferences—they either require a specific resource and derive a standard baseline benefit from it, or they find it entirely useless \citep{camacho2023}. By reducing the valuation space to these stark binary extremes, we can isolate the algorithmic friction caused by competing absolute needs, without the noise of marginal-utility differences.
\begin{definition}[Generalized Binary Valuation]
    A temporal fair division instance is under the generalized binary valuation setting iff \(\forall t \in [T], \forall g \in O_t, \forall i \in [n], v_i(g) \in \{0, b\}\) where \(b \in \mathbb{R}^+\).
\end{definition}

Building upon the dichotomy of generalized binary preferences, the third environment we investigate is the Bi-valued Valuation setting \citep{amanatidis2021}. In this domain, the valuation space is restricted to exactly two strictly positive outcomes, $a$ and $b$. While this retains the structural simplicity of polarized preferences, it captures a more nuanced economic reality. In many practical scenarios, agents possess strong, categorical preferences for high-tier resources, yet they still derive a tangible, baseline utility from secondary or less-preferred goods, rather than viewing them as entirely worthless. This setting allows us to examine the resilience of temporal fairness guarantees when all items are universally desirable, but subject to strict quality or preference tiers.
\begin{definition}[Bi-valued Valuation]
    A temporal fair division instance is under the bi-valued valuation setting iff \(\forall t \in [T], \forall g \in O_t, \forall i \in [n], v_i(g) \in \{b_1, b_2\}\) where \(b_1, b_2 \in \mathbb{R}^+\).
\end{definition}

The fourth environment we explore is the Identical Valuation setting \citep{plaut2017}. In this scenario, every agent assigns the exact same utility to any given item. This strict symmetry accurately models the distribution of universal commodities, standard-issue resources, or fungible assets, such as monetary grants, where subjective differences in preference are non-existent and allocations are purely driven by the objective value of the goods.
\begin{definition}[Identical Valuation]
    A temporal fair division instance is under the identical valuation setting iff \(\forall t \in [T], \forall g \in O_t, \forall i, j \in [n], v_i(g) = v_j(g)\).
\end{definition}

The final environment we explore is the House Allocation setting \citep{gan2019}. In this highly structured domain, the number of arriving goods exactly matches the number of agents at each time step, and the mechanism is strictly constrained to allocate exactly one item to every agent per round. This framework perfectly models assignment problems governed by rigid quotas or capacity constraints—such as assigning daily shifts to employees or allocating individual dormitories—and has been extensively studied to pursue economic efficiency \citep{dai2024}.
\begin{definition}[House Allocation]
    A temporal fair division instance is under the house allocation setting iff \(\forall t \in [T], |O_t| = n\), and \(\forall i \in N\), each agent must be allocated exactly one good.
\end{definition}
\vspace{-1em}
\section{Commonly Used Fair Division Algorithms}
In this section, we lay the methodological foundation for our work by detailing several prominent algorithmic frameworks in the fair division literature. Rather than merely listing these procedures, we examine their underlying mechanics—such as envy-cycle breaking, greedy sequential selection, and localized bundle swapping—to illustrate how they enforce strict equity constraints. We begin by reviewing classic offline mechanisms, namely Envy-Cycle Elimination (ECE), its strategic refinement MAX-ECE, and the Round-Robin (RR) protocol, which are celebrated for securing static fairness guarantees like EF1 and EFX. Subsequently, we transition to the dynamic setting by exploring Envy Balancing (EB), an algorithm explicitly designed to maintain temporal fairness (TEF1) across consecutive rounds. Together, these established protocols serve as the essential algorithmic building blocks for constructing and analyzing our novel sliding window solutions.

\subsection{Preliminaries}
A fair division instance is denoted by the tuple \(\mathcal{I} = \left( N, M, v = (v_1, \ldots, v_n)\right)\), where \(M\) is the set of items with \(m = |M|\) to be allocated to the set of agents \(N\) with \(n = |N|\) and valuation profile \(v\).

\subsection{Envy-Cycle Elimination}
Our first algorithmic building block is the classic \textbf{\textit{Envy-Cycle Elimination}} (ECE) algorithm introduced by \citet{lipton2004}. The core principle of ECE is to iteratively allocate the next available good to an \textit{unenvied} agent. Because the receiving agent is currently not envied by anyone, giving them one additional item ensures that any newly created envy is bounded by exactly that single item, thereby preserving the EF1 property.

However, during the allocation process, a state may arise where every agent is envied by at least one other agent. To resolve this and guarantee the existence of an unenvied agent, the algorithm maintains a directed \textit{envy-graph}. In this graph, vertices represent agents, and a directed edge from $i$ to $j$ exists if and only if agent $i$ envies agent $j$ (i.e., $v_i(A_i) < v_i(A_j)$). If every agent is envied, every node in the graph has an in-degree of at least one, meaning the graph must contain a directed cycle (an \textit{envy-cycle}).

The algorithm systematically breaks this cycle by performing a cyclic shift of bundles along the cycle's edges: each agent receives the bundle of the agent they currently envy. This exchange strictly improves the utility of every agent involved in the cycle without violating the existing EF1 guarantees. By repeatedly identifying and eliminating these cycles, the envy-graph is eventually reduced to a Directed Acyclic Graph (DAG). Every DAG must contain at least one source node (an agent with an in-degree of zero), which perfectly corresponds to an unenvied agent ready to receive the next good.

\begin{theorem}
    \label{thm:ECE1}
    ECE returns an EF1 allocation and runs in polynomial time \citep{lipton2004}.
\end{theorem}

Building upon this foundation, we consider a strategic refinement of the ECE framework known as \textbf{\textit{MAX-ECE}}. The primary distinction lies in the item selection rule: rather than allocating an arbitrary available good to the unenvied agent, MAX-ECE mandates a greedy approach where the agent receives their most preferred available good (i.e., the good that maximizes their marginal utility).

This deterministic selection mechanism inherently preserves the baseline EF1 guarantee established in Theorem~\ref{thm:ECE1}. More importantly, by ensuring that high-value items are distributed earlier in the allocation process, MAX-ECE yields strictly stronger fairness guarantees. As demonstrated in the subsequent theorems, this relatively simple modification allows the algorithm to achieve exact EFX for identical valuations, and a bounded approximation of EFX in the general case.

\begin{theorem}
    \label{thm:ECE2}
    MAX-ECE returns an EFX allocation for two agents if both agents have the same valuation profile \citep{plaut2017}.
\end{theorem}
\begin{theorem}
    \label{thm:ECE3}
    MAX-ECE returns a \(\frac{1}{2}\)-EFX allocation \citep{chan2019}.
\end{theorem}

\subsection{Round-Robin}
In addition to cycle-elimination techniques, we also analyze the \textbf{\textit{Round-Robin}} (RR) protocol. Widely recognized for its foundational role in job scheduling and operating system design \citep{panda2014}, RR has been extensively adapted in the fair division of indivisible goods due to its elegant simplicity and robust theoretical guarantees.

The mechanism operates by establishing a fixed, cyclic ordering over the set of $n$ agents. Allocation proceeds in sequential turns according to this predefined permutation. During their designated turn, an agent greedily selects their most preferred item from the pool of currently available goods (i.e., the item that maximizes their individual utility). Once every agent has made exactly one selection, a single "round" concludes; the sequence then resets to the first agent and repeats. This iterative, cyclic selection process continues until the set of unallocated goods is entirely exhausted. Because agents take turns sequentially, the maximum size difference between any two agents' bundles at any point during the execution is strictly bounded by one.

\begin{theorem}
    \label{thm:RR1}
    RR returns an EF1 allocation and runs in polynomial time \citep{aziz2021}.
\end{theorem}
\begin{observation}
    \label{thm:RR3}
    If agent \(i\) has just been allocated a good by RR, \(\forall j \leq i, |A_i| = |A_j|\) and \(\forall j > i, |A_i| = |A_j| + 1\).
\end{observation}

\section{Temporal Fair Division without scheduling} \label{woschedule}
\subsection{TEF1} \label{tef1}
The existence of TEF1 for goods under a general setting has been addressed by \cite{he2019}; thus, we focus on TEF1 under various other settings. Some restricted settings for goods have been studied by \cite{elkind2024}. In particular, they proposed polynomial-time algorithms for two types of goods under the generalized binary valuations setting, assuming agents have unimodal preferences over the goods. Additionally, they studied the case when multiple goods arrive in each round, and showed that when there are \(T = 2\) rounds, by running two Round-Robin algorithms, with one of the algorithms having its allocation sequence reversed, we achieve TEF1.

\subsubsection{Identical Days}
In the identical days setting with the house allocation setting, we can compute a TEF1 allocation in polynomial time when there are three rounds. We apply the result from \cite{elkind2024}, as mentioned above, to compute a TEF1 allocation in two rounds. We show that applying this result in rounds \(2\) and \(3\) in our setting yields a TEF1 allocation.

\begin{theorem}
    \label{thm:tef1identicaldaysHouse}
    A TEF1 allocation can be computed in polynomial time under the house allocation setting and identical days setting when \(T = 3\).
\end{theorem}
\begin{algorithm}
    \caption{Returns a TEF1 allocation}
    \label{alg:tef1identicaldaysHouse}
    \textbf{Input}: Arbitrary temporal fair division instance \(\mathcal{I} = \left( N, T, \left\{ O_t \right\}_{t \in [T]}, v = (v_1, \ldots, v_n), r \right)\)\\
    \textbf{Output}: TEF1 allocation \(\mathcal{A}\)
\begin{algorithmic}[1]
    \State $A^2 \gets$ RR$(N, O^2, v)$
    \For{$i = 1$ to $n$}
        \State $g_1 \gets A^2_i[1]$
        \State $g_2 \gets A^2_i[2]$
        \If{$g_1 \in A^1$ or $g_2 \in A^2$}
            \State $A^1_i \gets A^1_i \cup \{g_2\}$
        \Else
            \State $A^1_i \gets A^1_i \cup \{g_1\}$
        \EndIf
    \EndFor
    \State Reverse $N$
    \State $A^3 \gets A^3 \cup$ RR$(N, O_3, v)$
\end{algorithmic}
\end{algorithm}
\begin{proof}
    The polynomial runtime of Algorithm~\ref{alg:tef1identicaldaysHouse} is easy to verify, given that there is only one \textbf{For} loop, which runs in linear time, RR runs in polynomial time (Theorem~\ref{thm:RR1}), and the other operations run in polynomial time. Thus, we focus on proving correctness.
    
    Intuitively, we run RR in round \(2\), including the goods that arrive in round \(1\). Then, we run another RR, but with the allocation sequence reversed in round \(3\).
    
    In round \(2\), we run RR. Since it returns an EF1 allocation (Theorem~\ref{thm:RR1}), it remains to show that \(A^1\) is EF1 as well. Since in round \(1\), only one good arrives, it is trivially EF1. For each agent \(i\), \(A^2_i = \{g_{i, 1}, g_{i, 2}\}\), and one of the good should be from round \(1\). We show that there exists a way to select exactly one good from each agent, such that the set of these goods equals \(O_1\); thus, we select this good for each agent in round \(1\).
    
    Consider the following graph \(G = (V, E)\). The set of vertices \(V\) equals \(O^2\), and there exists an edge between any arbitrary two goods \(g_1\) and \(g_2\) in the edge set \(E\), if and only if \(\exists i \in N\) s.t. \(A^2_i = \{g_1, g_2\}\), or \(g_1\) is an identical copy of \(g_2\), which arrives on different rounds. We remove all repeated edges that occur when both of these conditions are satisfied.
    \begin{lemma}
        \label{lemma:tef1identicaldaysHouse1}
        \(G\) is bipartite if we consider identical copies on each side.
    \end{lemma}
    \begin{proof}
        Recall that \(G\) is bipartite if and only if no odd cycle occurs. Consider the special case where two goods are allocated to the same agent and are identical copies of each other. If we connect all these pairs of goods, the induced subgraph for all these vertices is bipartite. Notice that these two goods are not connected to any other vertices. Thus, we consider \(G\) with all of these vertices and edges removed in the remainder of our proof.
        
        For an arbitrary vertex \(g_1\), it is connected to two different vertices \(g_2\) and \(g_3\) respectively. Without loss of generality, assume \(\exists i \in N\) s.t. \(A^2_i = \{g_1, g_2\}\), and \(g_1\) is an identical copy of \(g_3\), which arrives on different rounds. \(g_2\) and \(g_3\) do not have any edge between them. Since each good only has exactly two identical copies, \(g_2\) and \(g_3\) are not identical copies of each other. In addition, \(g_3 \notin A^2_i\), so there does not exist any agent getting both \(g_2\) and \(g_3\).
        
        Then, we describe the construction for the connected component from \(g_1\). Let there be two sets \(L\) and \(R\). \(g_1\) is connected to \(g_2\) and \(g_3\) by assumption. We add \(g_1\) to \(L\), and \(g_2\), \(g_3\) to \(R\). Denote \(g_4\) to be a vertex that is the identical copy of \(g_2\). \(g_2\) is connected to \(g_1\) and \(g_4\). It can be proven similarly that \(g_1\) and \(g_4\) are not connected. \(g_4\) is added to \(L\). Then, \(g_4\) is connected to \(g_2\) and \(g_5\), a vertex that has an agent with both of these goods allocated to them. \(g_2\) and \(g_5\) are not connected. \(g_5\) is added to \(R\). It is easy to see that if we continue the construction, edges must exist between an element from \(L\) and an element from \(R\). Therefore, \(G\) is bipartite.
    \end{proof}
    By Lemma~\ref{lemma:tef1identicaldaysHouse1}, a bipartite graph for the goods exists; Algorithm~\ref{alg:tef1identicaldaysHouse} must return a valid solution.
    
    Additionally, Theorem 5.1 by \cite{elkind2024} states that by running RR with the allocation sequence (\(1, \dots, n\)) at the first round, and running RR with the allocation sequence (\(n, \dots, 1\)) at the second round, the final allocation is TEF1. This is equivalent to our case, as after round \(2\), our allocation equals the allocation returned by RR despite any changes in the allocation after round \(1\). Hence, our result follows.
\end{proof}

While our result establishes a definitive guarantee for $T = 3$ under the strict constraints of house allocation, the broader question of whether exact TEF1 allocations can be consistently guaranteed for general identical days instances extending beyond two rounds ($T > 2$) remains an open problem.

\subsection {TEFX} \label{tefx}
The existence of TEFX when there are at least two agents is determined to be negative by \cite{elkind2024}. They used a counterexample to show that even when there are two agents with identical valuations, with bi-valued goods, and \(T = 2\), a TEFX allocation may not exist. We then consider TEFX under various settings.

\subsubsection{Identical Days}
\begin{theorem}
\label{thm:tefxidenticaldays}
    A TEFX allocation is not guaranteed under the identical days setting when \(T > 2\), even if there are only two agents with identical valuations.
\end{theorem}
\begin{proof}
    Consider the instance with two agents and three goods that arrive in each round \(\forall t \in [T], O_t = \{g_1, g_2, g_3\}\), where agents have identical valuations: \(v(g_1) = 0, v(g_2) = 1\), and \(v(g_3) = 3\).
    \begin{lemma}
    \label{lemma:tefxidenticaldaysd1}
    On \(t = 1\), one agent must obtain the good valued \(0\) and \(1\), and the other agent must obtain the good valued \(3\) to achieve TEFX.
    \end{lemma}
    \begin{proof}
        Without loss of generality, assume the allocation of goods does not make agent \(2\) envy agent \(1\). Then, there are only three other cases.\\
        \textbf{Case 1:} Assume agent \(1\) does not get any good, and agent \(2\) gets the good valued \(0\), \(1\), and \(3\). Agent \(1\) envies agent \(2\), and the removal of any arbitrary good from agent \(2\), for example, the good valued \(0\), does not eliminate the envy of agent \(1\) towards agent \(2\). Therefore, we do not achieve TEFX in this case.\\
        \textbf{Case 2:} Assume agent \(1\) gets the good valued \(0\), and agent \(2\) gets the good valued \(1\) and \(3\). This case is similar to \textbf{Case 1}, where removing any good from agent \(2\) does not eliminate the envy from agent \(1\) towards agent \(2\). Therefore, we do not achieve TEFX in this case.\\
        \textbf{Case 3:} Assume agent \(1\) gets the good valued \(1\), and agent \(2\) gets the good valued \(0\) and \(3\). Agent \(1\) envies agent \(2\), and the removal of the good valued \(0\) from agent \(2\) does not eliminate the envy of agent \(1\) towards agent \(2\). Therefore, we do not achieve TEFX in this case.
        
        This finishes the proof.
\end{proof}
    \begin{table*}[h]
      \centering
      \caption{Lemma~\ref{lemma:tefxidenticaldaysd1} Different allocation illustration in round \(1\)}
      \label{tab:tefxidenticaldays1}
        \begin{tabular}{|c|c|c|}
          \hline
           & Agent 1 & Agent 2 \\
          \hline
          Successful & \(g_1\), \(g_2\) (\(v(A^1_1) = 1\)) & \(g_3\) (\(v(A^1_2) = 3\)) \\
          \hline
          Case 1 & \(\emptyset\) (\(v(A^1_1) = 0\)) & \(g_1\), \(g_2\), \(g_3\) (\(v(A^1_2) = 4\)) \\
          \hline
          Case 2 & \(g_1\) (\(v(A^1_1) = 0\)) & \(g_2\), \(g_3\) (\(v(A^1_2) = 4\)) \\
          \hline
          Case 3 & \(g_2\) (\(v(A^1_1) = 1\)) & \(g_1\), \(g_3\) (\(v(A^1_2) = 3\)) \\
          \hline
        \end{tabular}
    \end{table*}
    
    Without loss of generality, assume agent \(1\) gets the good valued \(0\) and \(1\), and agent \(2\) gets the good valued \(3\).
    \begin{lemma}
        \label{lemma:tefxidenticaldaysd2}
        On \(t = 2\), the good valued \(1\) must be allocated to agent \(2\), and the good valued \(3\) must be allocated to agent \(1\).
    \end{lemma}
    \begin{proof}
        \textbf{Case 1:} Assume agent \(1\) does not get any good, and agent \(2\) gets the good valued \(1\) and \(3\). Now, \(v_1(A_1) = 0 + 1 = 1\) and \(v_1(A_2) = 3 + 1 + 3 = 7\). It can be seen that \(v_1(A_1) < v_1(A_2 \backslash \{g\}), \forall g \in A_2\). Therefore, we do not achieve TEFX in this case.\\
        \textbf{Case 2:} Assume agent \(1\) gets the good valued \(1\), and agent \(2\) gets the good valued \(3\). Now, \(v_1(A_1) = 0 + 1 + 1 = 2\) and \(v_1(A_2) = 3 + 3 = 6\). It can be seen that \(v_1(A_1) < v_1(A_2 \backslash \{g\}), \forall g \in A_2\). Therefore, we do not achieve TEFX in this case.\\
        \textbf{Case 3:} Assume agent \(1\) gets the good valued \(1\) and \(3\), and agent \(2)\) does not get any good. Now, \(v_1(A_1) = 0 + 1 + 1 + 3 = 5\) and \(v_1(A_2) = 3\), so agent \(1\) does not envy agent \(2\), but agent \(2\) envies agent \(1\). Note that the removal of the good valued \(0\) in agent \(1\) does not eliminate the envy from agent \(2\) to agent \(1\). Therefore, we do not achieve TEFX in this case.
        
        This finishes the proof.
    \end{proof}
    \begin{table*}[h]
      \centering
      \caption{Lemma~\ref{lemma:tefxidenticaldaysd2} Different allocation illustration in round \(2\)}
      \label{tab:tefxidenticaldays2}
        \begin{tabular}{|c|c|c|}
          \hline
           & Agent 1 & Agent 2 \\
          \hline
          Successful & \(g_3\) (\(v(A^2_1) = 1 + 3 = 4\)) & \(g_2\) (\(v(A^2_2) = 3 + 1 = 4\)) \\
          \hline
          Case 1 & \(\emptyset\) (\(v(A^2_1) = 1\)) & \(g_2\), \(g_3\) (\(v(A^2_2) = 3 + 4 = 7\)) \\
          \hline
          Case 2 & \(g_2\) (\(v(A^2_1) = 1 + 1 = 2\)) & \(g_3\) (\(v(A^2_2) = 3 + 3 = 6\)) \\
          \hline
          Case 3 & \(g_2\), \(g_3\) (\(v(A^2_1) = 1 + 1 + 3 = 5\)) & \(\emptyset\) (\(v(A^2_2) = 3\)) \\
          \hline
        \end{tabular}
    \end{table*}
    \begin{proposition}
        \label{prop:tefxidenticaldays}
        The addition of a good where it is zero-valued for every agent in a previously non-TEFX allocation does not make it TEFX.
    \end{proposition}
    \begin{proof}
        If an allocation is TEFX, \(\forall j \in [N], v_i(A_i) \geq v_i(A_j \backslash \{g\}), \forall g \in A_j\). Assume agent \(i\) envies agent \(j\), and \(v_i(A_i) < v_i(A_j \backslash \{g\}), \forall g \in A_j\). If the good is allocated to agent \(i\),
        \[
        v_i(A_i) + 0 = v_i(A_i) < v_i(A_j \backslash \{g\}), \forall g \in A_j \text{.}
        \]
        Therefore, adding this zero-valued good does not convert a previously non-TEFX allocation into a TEFX.
    \end{proof}
    Hence, we first ignore the good valued \(0\) at \(t = 2\) and the good valued \(0\) at \(t = 3\), by showing that TEFX cannot be achieved at \(t = 3\) without both of these goods. Adding both goods back will not make us achieve TEFX either.
    \begin{lemma}
        \label{lemma:tefxidenticaldaysd3}
        TEFX cannot be achieved at \(t = 3\).
    \end{lemma}
    \begin{proof}
        \textbf{Case 1:} Agent \(1\) does not get any good, and agent \(2\) gets the good valued \(1\) and \(3\). Now, \(v_1(A_1) = 0 + 1 + 3 = 4\), and \(v_1(A_2) = 1 + 3 + 1 + 3 = 8\). Agent \(1\) still envies agent \(2\) after removing any arbitrary good from agent \(2\)'s allocation. Therefore, we do not achieve TEFX in this case.\\
        \textbf{Case 2:} Agent \(1\) gets the good valued \(1\), and agent \(2\) gets the good valued \(2\). Now, \(v_1(A_1) = 0 + 1 + 3 + 1 = 5\), and \(v_1(A_2) = 1 + 3 + 3 = 7\). The removal of the good valued \(1\) from agent \(2\) does not eliminate the envy from agent \(1\) towards agent \(2\). Therefore, we do not achieve TEFX in this case.\\
        \textbf{Case 3:} Agent \(1\) gets the good valued \(3\), and agent \(2\) gets the good valued \(1\). Now, \(v_1(A_1) = 0 + 1 + 3 + 3 = 7\), and \(v_1(A_2) = 1 + 3 + 1 = 5\). The removal of the good valued \(0\) from agent \(2\) does not eliminate the envy from agent \(2\) towards agent \(1\).\\
        \textbf{Case 4:} Agent \(1\) gets the good valued \(1\) and \(3\), and agent \(2\) does not get any good. Now, \(v_1(A_1) = 0 + 1 + 3 + 1 + 3 = 8\), and \(v_1(A_2) = 1 + 3 = 4\). It can be seen that no matter which good we remove from agent \(1\), agent \(2\) still envies agent \(1\). Therefore, we do not achieve TEFX in this case.
    \end{proof}
    \begin{table*}[h]
      \centering
      \caption{Lemma~\ref{lemma:tefxidenticaldaysd3} Different allocation illustration in round \(3\)}
      \label{tab:tefxidenticaldays3}
        \begin{tabular}{|c|c|c|}
          \hline
           & Agent 1 & Agent 2 \\
          \hline
          Case 1 & \(\emptyset\) (\(v(A^3_1) = 4\)) & \(g_2\), \(g_3\) (\(v(A^3_2) = 4 + 4 = 8\)) \\
          \hline
          Case 2 & \(g_2\) (\(v(A^3_1) = 4 + 1 = 5\)) & \(g_3\) (\(v(A^3_2) = 4 + 3 = 7\)) \\
          \hline
          Case 3 & \(g_3\) (\(v(A^3_1) = 4 + 3 = 7\)) & \(g_2\) (\(v(A^3_2) = 4 + 1 = 5\)) \\
          \hline
          Case 4 & \(g_2\), \(g_3\) (\(v(A^3_1) = 4 + 4 = 8\)) & \(\emptyset\) (\(v(A^3_2) = 4\)) \\
          \hline
        \end{tabular}
    \end{table*}
    Then, by Proposition~\ref{prop:tefxidenticaldays}, after adding the zero-valued goods, Lemma~\ref{lemma:tefxidenticaldaysd3} still holds true. Thus, our result follows.
\end{proof}

The next setting we study is the generalized binary valuation. This class of valuation functions generalizes both identical and binary valuations. We demonstrate that a TEFX allocation exists when there are two agents in this setting. Additionally, when there are two agents, EFX implies MMS. Therefore, we also find a TMMS allocation when there are two agents.

\subsubsection{Generalized Binary Valuation}
\begin{theorem}
    \label{thm:tefxgenbivalue}
    A TEFX and TMMS allocation can be computed in polynomial time when there are two agents under the generalized binary valuation setting.
\end{theorem}
\begin{algorithm}
    \caption{Returns a TEFX allocation under generalized binary valuation}
    \label{alg:tefxgenbivalue}
    \textbf{Input}: Arbitrary temporal fair division instance \(\mathcal{I} = \left( N, T, \left\{ O_t \right\}_{t \in [T]}, v = (v_1, \ldots, v_n), r \right)\)\\
    \textbf{Output}: TEFX allocation \(\mathcal{A}\)
\begin{algorithmic}[1]
    \State $f_1 \gets T + 1$
    \State $f_2 \gets T + 1$
    \For{$t = 1$ to $T$}
        \For{$g \in O_t$}
            \If{$v_1(g) = 1$ and $v_2(g) = 0$ and $f_1 = T + 1$}
                \State $f_1 \gets t$
            \EndIf
            \If{$v_1(g) = 0$ and $v_2(g) = 1$ and $f_2 = T + 1$}
                \State $f_2 \gets t$
            \EndIf
        \EndFor
    \EndFor
    \If{$f_1 < f_2$}
        \State $j \gets 2$
        \State $o \gets 1$
    \Else
        \State $j \gets 1$
        \State $o \gets 2$
    \EndIf
    \For{$t = 1$ to $T$}
        \For{$g \in O_t$}
            \If{$v_1(g) = v_2(g) = 0$}
                \State $A^t_o \gets A^t_o \cup \{g\}$
            \ElsIf{$v_1(g) = 0$ and $v_2(g) = b$}
                \State $A^t_2 \gets A^t_2 \cup \{g\}$
            \ElsIf{$v_1(g) = b$ and $v_2(g) = 0$}
                \State $A^t_1 \gets A^t_1 \cup \{g\}$
            \Else
                \State $A^t_j \gets A^t_j \cup \{g\}$
                \If {$j = 1$}
                    \State $j \gets 2$
                \Else
                    \State $j \gets 1$
                \EndIf
            \EndIf
        \EndFor
    \EndFor
\end{algorithmic}
\end{algorithm}
\begin{proof}
    The polynomial runtime of Algorithm~\ref{alg:tefxgenbivalue} is easy to verify, given that there are only nested \textbf{For} loops, which run in polynomial time, and the other operations run in polynomial time. Thus, we focus on proving correctness.
    
    Intuitively, Algorithm~\ref{alg:tefxgenbivalue} first allocates goods that are valued positively to both agents in a round-robin manner (allocate alternatively). Then, depending on which agent is first allocated a good that is valued positively by them but not by the other agent, we decide which agent to allocate the zero-valued good to. We do this to avoid a scenario where an agent envies another agent, but the latter holds a zero-valued good, so the envy cannot be eliminated if we remove this zero-valued good.
    
    Denote \(f_1\) and \(f_2\) to be the first round where agent \(1\) and \(2\) are allocated a good that is valued \(b\) for themselves, but \(0\) for the other agent. If this good does not appear for an agent \(i \in \{1, 2\}\), we define \(f_i = T + 1\). Without loss of generality, assume \(f_1 \leq f_2\). We split the proof for each case into three parts: round \(1\) to \(f_1 - 1\), round \(f_1\) to \(f_2 - 1\), round \(f_2\) to \(T\).\\
    \textbf{Case 1:} \(f_1 = f_2\).\\
    \textbf{Round \(1\) to \(f_1 - 1\):} Notice the part in Algorithm~\ref{alg:tefxgenbivalue} that allocates goods that are valued \(b\) for both agents is RR, where in this case, agent \(1\) is allocated first. By Observation~\ref{thm:RR3}, agent \(2\) will not be envied by agent \(1\) during this period. On the other hand, RR guarantees that for goods that are valued \(b\) for both agents, the difference in the number allocated for each agent will not be more than \(1\). Thus, agent \(2\) can drop its envy towards agent \(1\) by removing any good from their allocation, as \(v_2(A^t_2) = v_2(A^t_1)\) or \(v_2(A^t_2) = v_2(A^t_1) - b\).\\
    \textbf{Round \(f_1\) to \(f_2 - 1\):} This period of rounds is invalid.\\
    \textbf{Round \(f_2\) to \(T\):} Let \(P\) and \(Q\) be a set of goods such that \(\forall p \in P\), and \(\forall q \in Q\), \(v_1(p) = v_2(q) = b\), and \(v_1(q) = v_2(p) = 0\). Note that \(\forall f_2 \leq t \geq T, v_1(A^t_2 \backslash Q) = v_1(A^t_2)\), and \(v_2(A^t_1 \backslash P) = v_2(A^t_1)\). Therefore, this period is equivalent to the period of round \(1\) to \(f_1 - 1\).\\
    \textbf{Case 2:} \(f_1 < f_2\).\\
    \textbf{Round \(1\) to \(f_1 - 1\):} The proof for this part is the same as the same period in \textbf{Case 1}, except that agent \(2\) is allocated first.\\
    \textbf{Round \(f_1\) to \(f_2 - 1\):} Let \(P\) be a set of goods such that \(\forall p \in P\), \(v_1(p) = b\), and \(v_2(p) = 0\). Note that \(\forall f_2 \leq t \geq T, v_2(A^t_1 \backslash P) = v_2(A^t_1)\). Therefore, this period is equivalent to the period of round \(1\) to \(f_1 - 1\).\\
    \textbf{Round \(f_2\) to \(T\):} This period is the same as \textbf{Case 1}.\\
    Since EFX implies MMS for two agents, our result follows.
\end{proof}

It was shown that TEFX allocation may not exist when all agents have identical valuations \citep{elkind2024}. We combine the identical valuation setting with the generalized binary valuation setting, and establish a strong positive result by modifying the Round-Robin algorithm.
\begin{theorem}
\label{thm:tefxgenbivalueidentical}
A TEFX allocation can be computed in polynomial time when all agents' valuations are both generalized binary and identical.
\end{theorem}
\begin{algorithm}
    \caption{Returns a TEFX allocation under generalized binary valuation and identical valuations}
    \label{alg:tefxgenbivalueidentical}
    \textbf{Input}: Arbitrary temporal fair division instance \(\mathcal{I} = \left( N, T, \left\{ O_t \right\}_{t \in [T]}, v = (v_1, \ldots, v_n), r \right)\)\\
    \textbf{Output}: TEFX allocation \(\mathcal{A}\)
\begin{algorithmic}[1]
    \State $j \gets 1$
    \For{$t = 1$ to $T$}
        \For{$g \in O_t$}
            \If{$v_1(g) = 0$}
                \State $A^t_n \gets A^{t-1}_n \cup g$
            \Else
                \State $A^t_j \gets A^{t-1}_j \cup g$
                \State $j \gets j + 1$
                \If{$j > n$}
                    \State $j \gets 1$
                \EndIf
            \EndIf
        \EndFor
    \EndFor
\end{algorithmic}
\end{algorithm}

\begin{proof}
It is easy to observe that Algorithm~\ref{alg:tefxgenbivalueidentical} runs in polynomial time, as there are only nested \textbf{For} loops, which run in polynomial time, and the other operations run in polynomial time. Therefore, we focus on proving its correctness.

Please note that this is RR with some modifications when handling zero-valued goods. By Observation~\ref{thm:RR3}, \(\forall j < i, v_i(A_i) = v_i(A_j)\). Therefore we now focus on agents \(j > i\).\\
\textbf{Case 1:} \(j \neq n\). In this case, \(|A_i^t| = |A_j^t| + 1\), thus \(v_j(A_j) \geq v_j(A_j) - v_j(g) = v_i(A_i) - v_i(g), \forall g \in O_t\).\\
\textbf{Case 2:} \(j = n\). In this case, since it is possible that \(\exists g \in A_n^t\) s.t. \(v_i(g) = 0\), agent \(i\) must not envy agent \(n\). However, since \(\forall i < n, |A_i^t| = |A_n^t| + 1\), 
\[v_i(A_i) \geq v_i(A_n) \geq v_i(A_n \backslash \{g\}), \forall g \in A_n \text{.}
\]
Therefore, Algorithm~\ref{alg:tefxgenbivalueidentical} returns a TEFX allocation.
\end{proof}

In summary, our exploration of exact TEFX reveals a stark reality for temporal fair division: maintaining "up to any good" fairness round by round is overwhelmingly impossible in dynamic environments. While we successfully identified narrow theoretical sanctuaries, specifically, systems constrained to generalized binary valuations coupled with either a strict two-agent limit or perfectly identical preferences, these highly structured conditions are the exception rather than the rule.

The widespread failure of exact TEFX under broader constraints, including the predictable symmetry of identical days or multi-agent binary settings, demonstrates that cumulative, exact fairness is fundamentally incompatible with the inherent friction of online arrivals. Because real-world allocation problems rarely conform to such rigid mathematical boundaries, the relentless impossibilities surrounding TEFX strongly motivate relaxing our fairness criteria. This naturally paves the way for our subsequent analysis, where we abandon the fragile pursuit of exact equity and instead establish robust, multiplicative fairness guarantees through the lens of $\alpha$-TEFX.

\subsection {$\alpha$-TEFX} \label{atefx}
In this section, we consider the approximation variant of TEFX, \(\alpha\)-TEFX. We can provide positive results in a general setting and achieve stronger results under various settings. We note that the main challenge for obtaining an \(\alpha\)-TEFX allocation is when there exist some goods that are zero-valued for some agents. Unlike online \(\alpha\)-EFX, which were studied previously \citep{melissourgos2025}, in \(\alpha\)-TEFX, we focus on avoiding cases such as the following: after a round, an agent \(i\) that has no goods allocated to them, envies another agent, that has one good positively valued for \(i\), and another good that is zero-valued for \(i\). This is acceptable for standard online \(\alpha\)-EFX, but results in an invalid approximate TEFX allocation.

\subsubsection{General Settings}
We first demonstrate that the approximation ratio must depend on the values of the goods.

\begin{theorem}
    \label{thm:atefxconst}
    There does not exist any algorithm that returns a positive constant approximation for \(\alpha\)-TEFX, even when there are only two agents with identical valuations.
\end{theorem}
\begin{proof}
     Assume there exists a solution that computes a \(k\)-TEFX allocation. Consider the instance with two agents and three goods \(O = \{g_1, g_2, g_3\}\), where agents have identical valuations: \(v(g_1) = v(g_2) = 1\) and \(v(g_3) = 
     \frac{1}{k} + 1\), where \(k > 0\) and is a constant. Let \(g_i\) arrive in round \(i\). In order to obtain a \(k\)-TEFX allocation at the first round, \(g_1\) and \(g_2\) must be allocated to each agent, respectively. Then, no matter which agent is allocated \(g_3\), we do not obtain a \(k\)-TEFX allocation. Without loss of generality, let it be agent \(1\), we have
    \[
    v(g_2) = 1 < k \cdot (v(g_1 \cup g_3) - v(g_1)) = k \cdot (\frac{1}{k} + 1) = k + 1 \text{,}
    \]
    which contradicts our assumption.
\end{proof}

As a result, we observe that the approximation ratio should be expressed in terms of the minimum- and maximum-valued goods. Therefore, we propose a loose bound for approximating TEFX when all goods are valued positively by all agents. Intuitively, we compute a $1/2$-EFX allocation in each round and add it to the allocation.

\begin{theorem}
    \label{thm:atefx}
    A \(\frac{\frac{1}{2} \cdot min_{g \in O} v_i(g)}{min_{g \in O} v_i(g) + \frac{1}{2} \cdot max_{g \in O} v_i(g)}\)-TEFX allocation can be computed in polynomial time, if all agents value all goods positively.
\end{theorem}
\begin{algorithm}
    \caption{Returns a \(\frac{\frac{1}{2} \cdot min_{g \in O} v_i(g)}{min_{g \in O} v_i(g) + \frac{1}{2} \cdot max_{g \in O} v_i(g)}\)-TEFX allocation}
    \label{alg:atefx}
    \textbf{Input}: Arbitrary temporal fair division instance \(\mathcal{I} = \left( N, T, \left\{ O_t \right\}_{t \in [T]}, v = (v_1, \ldots, v_n), r \right)\)\\
    \textbf{Output}: \(\frac{\frac{1}{2} \cdot min_{g \in O} v_i(g)}{min_{g \in O} v_i(g) + \frac{1}{2} \cdot max_{g \in O} v_i(g)}\)-TEFX allocation \(\mathcal{A}\)
\begin{algorithmic}[1]
    \State $missing \gets n$
    \For{$t = 1$ to $T$}
        \If{$|O_t| < missing$}
            \State $m \gets |O_t|$
            \For{$l = 1$ to $m$}
                \While{$\forall i \in N, \exists j \in N$ s.t. $v_j(A^t_j) < v_j(A^t_i)$}
                    \State Resolve the cycle by exchanging \(A^t_i\) where \(i\) is the agent.
                \EndWhile
                \State Let $i$ be an unenvied agent
                \State Let $g^* \in arg max_{g \in O_t} v_i(g)$
                \State $A^t_i \gets A^t_i \cup \{g^*\}$
                \State $O_t \gets O_t \backslash \{g^*\}$
            \EndFor
            \State $missing \gets missing - |O_t|$
        \ElsIf{$missing > 0$}
            \State $m \gets |O_t|$
            \For{$l = 1$ to $m$}
                \While{$\forall i \in N, \exists j \in N$ s.t. $v_j(A^t_j) < v_j(A^t_i)$}
                    \State Resolve the cycle by exchanging \(A^t_i\) where \(i\) is the agent.
                \EndWhile
                \State Let $i$ be an unenvied agent
                \State Let $g^* \in arg max_{g \in O_t} v_i(g)$
                \State $A^t_i \gets A^t_i \cup \{g^*\}$
                \State $O_t \gets O_t \backslash \{g^*\}$
            \EndFor
            \State $missing \gets 0$
            \State $m \gets |O_t|$
            \State $A^* \gets A^t$
            \For{$l = 1$ to $m$}
                \While{$\forall i \in N, \exists j \in N$ s.t. $v_j(A^t_j \backslash A^*_j) < v_j(A^t_i \backslash A^*_i)$}
                    \State Resolve the cycle by exchanging \(A^t_i \backslash A^*_i\) where \(i\) is the agent.
                \EndWhile
                \State Let $i$ be an unenvied agent
                \State Let $g^* \in arg max_{g \in O_t} v_i(g)$
                \State $A^t_i \gets A^t_i \cup \{g^*\}$
                \State $O_t \gets O_t \backslash \{g^*\}$
            \EndFor
        \Else
            \State $m \gets |O_t|$
            \For{$l = 1$ to $m$}
                \While{$\forall i \in N, \exists j \in N$ s.t. $v_j(A^t_j \backslash A^{t-1}_j) < v_j(A^t_i \backslash A^{t-1}_i)$}
                    \State Resolve the cycle by exchanging \(A^t_i \backslash A^{t-1}_i\) where \(i\) is the agent.
                \EndWhile
                \State Let $i$ be an unenvied agent
                \State Let $g^* \in arg max_{g \in O_t} v_i(g)$
                \State $A^t_i \gets A^t_i \cup \{g^*\}$
                \State $O_t \gets O_t \backslash \{g^*\}$
            \EndFor
        \EndIf
    \EndFor
\end{algorithmic}
\end{algorithm}
\begin{proof}
    The polynomial runtime of Algorithm~\ref{alg:atefx} is easy to verify, given that there are only nested \textbf{For} loops, which run in polynomial time, the slightly modified MAX-ECE runs in polynomial time (Theorem~\ref{thm:ECE1}), and the other operations run in polynomial time. Thus, we focus on proving correctness.
    
    We use MAX-ECE, which returns a \(\frac{1}{2}\)-EFX allocation (Theorem~\ref{thm:ECE3}). Intuitively, we consider the first \(n\) goods as a single round and run MAX-ECE. This is an EFX allocation because each agent gets exactly one good, which must be valued positively for all agents, and the removal of this good eliminates any envy towards any agent. Then, for each remaining round, we compute using MAX-ECE, and the final allocation is simply the union of all previous rounds' allocations. It is noted that we cannot use the Draft-and-Eliminate algorithm by \cite{amanatidis2020}, which returns a (\(\phi - 1\))-EFX allocation, as it cannot guarantee that each agent is first allocated exactly one good, and if we forcefully do so, it may break the approximation guarantee.
    
    The following lemma shows a lower bound of approximate-TEFX by proving the lower bound on the stronger notation, TEF.
    \begin{lemma}
        \label{lemma:atefx1}
            If we can compute an \(\alpha\)-EF allocation in each round, after \(T\) rounds we obtain an \(\alpha\)-TEF allocation.
    \end{lemma}
    \begin{proof}
        If we compute a \(\alpha\)-EFX allocation, for each agent \(i \in N\), we have
        \[
        v_i(A_i) \geq \alpha \cdot v_i(A_j) \text{.}
        \]
        If we union all allocation from round \(1\) to round \(t\), we have
        \[
        v_i(A^t_i) \geq \sum_{t' \in [t]} \alpha \cdot v_i(A^{t'}_j) = \alpha \cdot v_i(A^t_j) \text{.}
        \]
    \end{proof}
    We need to convert the approximation ratio from EFX to EF in order to apply the result from Lemma~\ref{lemma:atefx1}.
    
    Denote \(\beta = \frac{1}{2}\). We have
    \[
    v_i(A_i) \geq \alpha \cdot v_i(A_j \backslash \{g\}), \forall g \in A_j \text{.}
    \]
    Denote \(\gamma = 1 + \frac{\beta \cdot v_i(g)}{v_i(A_i)}\). Notice that
    \[
    \gamma \cdot v_i(A_i) \geq \beta \cdot v_i(A_j) \text{.}
    \]
    Therefore, MAX-ECE returns a \(\frac{\beta}{\gamma}\)-TEF allocation.
    \begin{lemma}
        \label{lemma:atefx2}
        The minimum value of \(\frac{\beta}{\gamma}\) is when \(v_i(A_i) = min_{g \in O} v_i(g)\) and \(v_i(g) = max_{g \in O, v_i(g)} v_i(g)\).
    \end{lemma}
    \begin{proof}
    Denote \(x = v_i(A_i)\), \(y = v_i(g)\), and
    \[
    f(x, y) = \frac{\beta}{\gamma} = \frac{\frac{1}{2} \cdot x}{x + \frac{1}{2} \cdot y} \text{.}
    \]
    Recall that the critical points of \(f\) occur at
    \[
    \nabla f(x, y) = \left(\frac{\partial f}{\partial x}, \frac{\partial f}{\partial y}\right) = \left(0, 0\right) \text{.}
    \]
    By solving this equation, we find that \(f\) does not have any critical points. By parameterizing \(f\) by the slope \(t = \frac{y}{x}\), then \(f(t) = \frac{1}{2 + t}\), which is strictly decreasing with respect to \(t\). Therefore, in order to maximize \(t\), we maximize \(y\) and minimize \(x\).
    Since
    \[
    v_i(A_i) \in [min_{g \in O} v_i(g), \sum_{g \in O} v_i(g)]
    \]
    and
    \[
    v_i(g) \in [min_{g \in O} v_i(g), max_{g \in O, v_i(g)} v_i(g)] \text{,}
    \]
    the minimum value for \(\frac{\beta}{\gamma}\) is when \(v_i(A_i) = min_{g \in O} v_i(g)\) and \(v_i(g) = max_{g \in O, v_i(g)} v_i(g)\).
\end{proof}
    Therefore, after substituting \(v_i(A_i)\) and \(v_i(g)\), and by Lemma~\ref{lemma:atefx1}, the approximation ratio of TEFX by Algorithm~\ref{alg:atefx} is
    \[
    \frac{\beta}{\gamma} = \frac{\frac{1}{2} \cdot min_{g \in O} v_i(g)}{min_{g \in O} v_i(g) + \frac{1}{2} \cdot max_{g \in O} v_i(g)} \text{.}
    \]
\end{proof}

Under the identical days setting, we propose a stronger approximation guarantee for TEFX with two agents using ECE based on Theorem~\ref{thm:ECE3}.

\subsubsection{Identical Days}
\begin{theorem}
    \label{thm:atefxidenticaldays}
    A \(\frac{1}{2}\)-TEFX allocation can be computed in polynomial time when there are two agents under the identical days setting.
\end{theorem}
\begin{algorithm}
    \caption{Returns a \(\frac{1}{2}\)-TEFX allocation under the identical days setting when \(n = 2\)}
    \label{alg:atefxidenticaldays}
    \textbf{Input}: Arbitrary temporal fair division instance \(\mathcal{I} = \left( N, T, \left\{ O_t \right\}_{t \in [T]}, v = (v_1, \ldots, v_n), r \right)\)\\
    \textbf{Output}: \(\frac{1}{2}\)-TEFX allocation \(\mathcal{A}\)
\begin{algorithmic}[1]
    \For{$t = 1$ to $T$}
        \If{$t$ is odd}
            \State $P, Q \gets$ MAX-ECE$(N, O_t, \{v_1, v_1\})$
            \If{$v_2(P) \geq v_2(Q)$}
                \State $A^t_1, A^t_2 \gets Q, P$
            \Else
                \State $A^t_1, A^t_2 \gets P, Q$
            \EndIf
        \Else
            \State $P, Q \gets$ MAX-ECE$(N, O_t, \{v_2, v_2\})$
            \If{$v_1(P) \geq v_1(Q)$}
                \State $A^t_1, A^t_2 \gets P, Q$
            \Else
                \State $A^t_1, A^t_2 \gets Q, P$
            \EndIf
        \EndIf
    \EndFor
\end{algorithmic}
\end{algorithm}
\begin{proof}
    The polynomial runtime of Algorithm~\ref{alg:atefxidenticaldays} is easy to verify, given that there is only one \textbf{For} loop, which runs in linear time, MAX-ECE runs in polynomial time (Theorem~\ref{thm:ECE1}), and the other operations run in polynomial time. Thus, we focus on proving correctness.
    
    MAX-ECE computes an EFX allocation when there are two agents (Theorem~\ref{thm:ECE2}). Intuitively, we first compute an EFX allocation using MAX-ECE, assuming that both agents have the valuation profile of agent \(1\) (resp. agent \(2\)). Then, agent \(2\) (resp. agent \(1\)) selects the higher valued bundle; thus, this agent does not envy the other bundle. Lastly, agent \(1\) (resp. agent \(2\)) selects the other bundle, which satisfies EFX.
    
    First, we address the special case where in each round, only one good arrives. This case is trivial, as each agent is allocated a copy of the identical good alternatively, and we obtain a TEFX allocation. Thus, we consider the case where, in each round, more than two goods arrive. Assume we have allocated \(t\) rounds of items. By design, the algorithm allocates the ``same bundle" to the same agent in odd and even rounds, respectively. Denote the allocation that is allocated to agent \(i\) in odd rounds be \(A^{odd}_i\), and in even rounds be \(A^{even}_i\).\\
    \textbf{Case 1:} \(t\) is even. Let \(k = \frac{t}{2}\). First, for even rounds \(t_{even}\), we have
    \[
    v_1(A^{even}_1) \geq v_1(A^{even}_2 \backslash \{g\}),  g = arg min_{g' \in A^{even}_2} v_1(g') \text{.}
    \]
    In addition, the value of \(g\) can be described as
    \[
    v_1(g) \leq v_1(A^{even}_2) \cdot \frac{1}{2} \text{.}
    \]
    If \(v_1(g) > v_1(A^{even}_2) \cdot \frac{1}{2}\), it either means \(A^{even}_2 = \{g\}\) or there exists a good \(g^*\) s.t. \(v_1(g) > v_1(g^*)\). For the first case, it is not possible, as this is an even round; therefore, at least one round has passed before this, which implies agent \(2\) must be allocated at least two goods. For the second case, in order to maintain EFX in MAX-ECE, we must allocate at least one good to each agent. There exists a good \(g^*\) s.t. \(v_1(g) > v_1(g^*)\), which means \(g\) is not the minimum valued good for agent \(1\) in the allocation \(A^{even}_2\).
    
    For odd rounds \(t_{odd}\), we have
    \[
    v_i(A^{odd}_1) \geq v_1(A^{odd}_2) \text{.}
    \]
    Therefore, we have
    \[
    k \cdot v_1(A^{odd}_1 \cup A^{even}_1) \geq k \cdot v_1(A^{odd}_2 \cup A^{even}_2) - k \cdot v_1(g) \text{.}
    \]
    By substituting \(v_1(A^{even}_2) \cdot \frac{1}{2}\) into \(v_1(g)\) through the inequality and simplifying the expression, we have
    \[
    v_1(A^{odd}_1 \cup A^{even}_1) \geq \frac{1}{2} \cdot v_1(A^{odd}_2 \cup A^{even}_2) \text{.}
    \]
    \textbf{Case 2:} \(t\) is odd. Notice that since \(v_i(A^{odd}_1) \geq v_1(A^{odd}_2)\), we obtain a result not worse than \textbf{Case 1}, which means it is also \(\frac{1}{2}\)-TEFX.\\
    A similar proof can be obtained for agent \(2\).
\end{proof}

We also propose a stronger approximation guarantee for TEFX under the generalized binary valuation setting.

\subsubsection{Generalized Binary Valuation}
\begin{theorem}
    \label{thm:atefxgenbivaluation}
    Algorithm~\ref{alg:atefxgenbivaluation} returns a \(\frac{1}{2}\)-TEFX allocation in polynomial time for generalized binary valuation.
\end{theorem}
\begin{proof}
    The polynomial runtime of Algorithm~\ref{alg:atefxgenbivaluation} is easy to verify, given that there are only nested \textbf{For} loops, which run in polynomial time, and the other operations run in polynomial time. Thus, we focus on proving correctness.
    
    Algorithm~\ref{alg:atefxgenbivaluation} is a modified version of Algorithm 3 from \cite{elkind2024}, which returns a TEF1 allocation under the generalized binary valuation setting.
    
    Intuitively, for goods with at least one agent who values them positively, we follow the original algorithm, allocating to the agent with the minimum-valued allocation, with ties broken arbitrarily. Then, for goods that are valued at zero for all agents, we allocate them to the agent that is allocated the positively valued good lastly, or to any agent that is not allocated any good at all. Denote this agent as agent \(last\). We note that this modification does not affect the TEF1 property, as in the original algorithm, the zero-valued good can be allocated to any agent.
    
    We first show a key property for Algorithm~\ref{alg:atefxgenbivaluation}.
    \begin{lemma}
        \label{lemma:atefxgenbivaluation2}
        If agent \(i\) is allocated good \(g\) in round \(t\), and agent \(j\) envies agent \(i\) after the allocation of this good, either agent \(j\)'s allocation has at least one positively valued good for themselves, or agent \(i\) does not have any zero-valued good to agent \(j\).
    \end{lemma}
    \begin{algorithm}
    \caption{Returns a \(\frac{1}{2}\)-TEFX allocation under the generalized binary valuation setting}
    \label{alg:atefxgenbivaluation}
    \textbf{Input}: Arbitrary temporal fair division instance \(\mathcal{I} = \left( N, T, \left\{ O_t \right\}_{t \in [T]}, v = (v_1, \ldots, v_n), r \right)\)\\
    \textbf{Output}: \(\frac{1}{2}\)-TEFX allocation \(\mathcal{A}\)
\begin{algorithmic}[1]
    \State $last \gets -1$
    \For{$t = 1$ to $T$}
        \For{$g \in O_t$}
            \State $S \gets \{i^* \in N | v_{i^*}(g) > 0\}$
            \If{$S \neq \emptyset$}
                \State $i \in arg min_{i^* \in S} v_{i^*}(A^t_{i^*})$, with ties broken arbitrarily
                \If{$A^t_i = \emptyset$}
                    \State $last \gets i$
                \EndIf
                \State $A^t_i \gets A^t_i \cup \{g\}$
                \State $O_t \gets O_t \backslash \{g\}$
            \EndIf
        \EndFor
    \EndFor
    \For{$i \in N$}
        \If{$A_i = \emptyset$}
            \State $last \gets i$
        \EndIf
    \EndFor
    \For{$t = 1$ to $T$}
        \For{$g \in O_t$}
            \State $A^t_{last} \gets A^t_{last} \cup \{g\}$
        \EndFor
    \EndFor
\end{algorithmic}
\end{algorithm}
    First, we show that before agent \(last\) is allocated any good, the allocation achieves TEFX. To simplify the proof, we order the agents by the time they are first allocated any good, with the exception of agent \(last\), who is placed last. We ignore the agents that are not allocated any goods, as they are equivalent to agent \(last\). Assume we have allocated good \(g\) to agent \(i\). By the design of Algorithm~\ref{alg:atefxgenbivaluation}, if \(j < i\), agent \(j\) is allocated exactly one good. Therefore, agent \(i\) does not envy any agent \(j < i\). If \(i < j \leq last\), agent \(j\) may envy agent \(i\). 
    By Lemma~\ref{lemma:atefxgenbivaluation2}, if agent \(j\) envies agent \(i\), agent \(i\) does not have any zero-valued good to agent \(j\). Since agent \(i\) is allocated exactly one good, \(v_j(A^t_i \backslash \{g\}) = 0, \forall g \in A^t_i\), which implies \(v_j(A^t_j) = 0 \geq v_j(A^t_i \backslash \{g\}), \forall g \in A^t_i\).
    After agent \(last\) is allocated a good that is positively valued to themselves, all agents have at least one positively valued good for themselves, which means \(v_i(A^t_i) \geq b, \forall i \in N, \forall t \in [T]\). Then, since Algorithm~\ref{alg:atefxgenbivaluation} returns TEF1 allocation, \(\forall t \in [T]\), we have
    \(v_i(A^t_i) \geq v_i(A^t_j)\) or \(v_i(A^t_j) - v_i(A^t_i) = b\). The former case is trivially TEFX; therefore, we consider the latter case. Notice that \(\forall i, j \in N, v_i(A_j) = k \cdot b\) where \(k \in \mathbb{N}\). Let \(v_i(A^t_i) = k_1 \cdot b\), and \(v_i(A^t_j) = k_2 \cdot b, k_2 \geq 2\). Therefore, \(k_1 = k_2 - 1\). Thus, we have
    \[
    \frac{k_1 \cdot b}{k_2 \cdot b} = \frac{k_2 - 1}{k_2} \geq \frac{1}{2} \;\Rightarrow\; v_i(A^t_i) \geq \frac{1}{2} \cdot v_i(A^t_j) \text{.}
    \]
\end{proof}
As suggested in Theorem~\ref{thm:atefxconst}, under the identical valuation setting, a positive constant approximation ratio does not exist. We can propose a stronger result than the general setting. Intuitively, we handle the zero-valued good carefully and always allocate the incoming good in each round to the agent with the least-valued bundle.
\subsubsection{Identical Valuation}
\begin{theorem}
\label{thm:atefxidenticalvaluation}
    A \(\frac{min_{g \in O, v(g) > 0} v(g)}{max_{g \in O} v(g) + min_{g \in O, v(g) > 0} v(g)}\)-TEFX allocation can be computed in polynomial time under the identical valuation setting.
\end{theorem}
\begin{algorithm}
    \caption{Returns a \(\frac{min_{g \in O, v(g) > 0} v(g)}{max_{g \in O} v(g) + min_{g \in O, v(g) > 0} v(g)}\)-TEFX allocation under the identical valuation setting}
    \label{alg:atefxidenticalvaluation}
    \textbf{Input}: Arbitrary temporal fair division instance \(\mathcal{I} = \left( N, T, \left\{ O_t \right\}_{t \in [T]}, v = (v_1, \ldots, v_n), r \right)\)\\
    \textbf{Output}: \(\frac{1}{max_{g \in O} v(g) + 1}\)-TEFX allocation \(\mathcal{A}\)
\begin{algorithmic}[1]
    \For{$t = 1$ to $T$}
        \For{$g \in O_t$}
            \If{$v(g) = 0$}
                \State $A^t_n \gets A^t_n \cup \{g\}$
            \Else
                \For{$i = 1$ to $n$}
                \If{$i = min_{j \in N} v(A^t_j)$}
                    \State $A^t_i \gets A^t_i \cup \{g\}$
                \EndIf
            \EndFor
            \EndIf
        \EndFor
    \EndFor
\end{algorithmic}
\end{algorithm}
\begin{proof}
    The polynomial runtime of Algorithm~\ref{alg:atefxidenticalvaluation} is easy to verify, given that there are only nested \textbf{For} loops, which run in polynomial time, and the other operations run in polynomial time. Thus, we focus on proving correctness.
    
    Intuitively, Algorithm~\ref{alg:atefxidenticalvaluation} allocates the goods to the least valued allocation among all agents. If there are multiple agents with the least-valued allocation, we allocate to the agent with the smaller index. Lastly, we always allocate zero-valued goods to agent \(n\).

    Let us first address the special case. If \(\forall g \in O, v(g) = 0\), no matter how we allocate the goods, we result in a TEFX allocation. Therefore, we assume \(\exists g \in O\) s.t. \(v(g) > 0\).
    
    Denote \(g_{max} = arg max_{g \in O} v(g)\).
    \begin{lemma}
        \label{lemma:atefxidenticalvaluation1}
        At all times, \(\forall i, j \in [n], |v(A^t_i) - v(A^t_j)| \leq v(g_{max})\).
    \end{lemma}
    \begin{proof}
        We prove by contradiction. Assume at some arbitrary time, \(\forall i, j \in [n], |v(A^t_i) - v(A^t_j)| \leq v(g_{max})\). After the allocation of good \(g^*\) to agent \(i\), \(\exists j \in [n]\) s.t. \(|v(A^t_i) - v(A^t_j)| > v(g_{max})\). Before the allocation of \(g^*\), if \(v(A^t_i) > v(A^t_j)\), then the good will not be allocated to agent \(i\), thus \(v(A^t_i) < v(A^t_j)\). Therefore, we have
        \[
        v(A^t_j) - v(A^t_i) \leq v(g_{max})
        \]
        before \(g^*\) is allocated to agent \(i\). If \(v(g^*) \leq v(g_{max})\), 
        \[
        v(A_i) - v(A_j) + v(g^*) \leq v(g_{max}) \text{,}
        \]
        since \(v(A_i) - v(A_j) \geq -v(g_{max})\), which contradicts our assumption. Thus, \(v(g^*) > v(g_{max})\). However, this means \(g_{max} \neq arg max_{g \in O} v(g)\), which contradicts our definition of \(g_{max}\).
    \end{proof}
    By Lemma~\ref{lemma:atefxidenticalvaluation1}, we know that the difference between any two agents at any time must be less than or equal to the maximum valued good. Without loss of generality, assume \(v(A^t_i) \geq v(A^t_j) > 0\). Denote \(d = v(A^t_i) - v(A^t_j)\). Let \(v(A^t_j) \geq k \cdot v(A^t_i)\) where \(k > 0\). Then, we have
    \[
    k \leq \frac{v(A^t_j)}{d + v(A^t_j)} \text{.}
    \]
    \(k\) attains its maximum value when it equals \(\frac{v(A^t_j)}{d + v(A^t_j)}\). Denote \(g^* = arg min_{g \in O} v(g)\). \(k \geq \frac{v(g^*)}{d + v(g^*)}\), since \(v(A^t_j) \geq v(g^*)\). Thus, we have proven Theorem~\ref{thm:atefxidenticalvaluation} when \(v(A^t_j) > 0\). Lastly, we consider the case when \(v(A^t_j) = 0\).
    \begin{lemma}
    \label{lemma:atefxidenticalvaluation2}
        We have a EFX allocation when \(\exists i \in [n], v(A^t_i) = 0\).
    \end{lemma}
    \begin{proof}
        Notice that Algorithm~\ref{alg:atefxidenticalvaluation} first allocates exactly one positively valued good to each agent. More precisely, from agent \(1\) to agent \(n - 1\), each agent is allocated exactly one good, and agent \(n\) is allocated exactly one positively valued good, and multiple, or possibly none, zero-valued good(s). This means for all agents \(i \in [n]\), \(\forall j \in [n - 1]\),
        \[
        v(A^t_i) \geq v(A^t_j \backslash \{g\}) = v(\emptyset) = 0, \forall g \in A^t_j \text{.}
        \]
        Lastly, since we allocate the positively valued goods with respect to the index of the agents, before the allocation of the first positively valued good to agent \(n\), no agent will envy agent \(n\). After the allocation of this good, \(\forall i \in [n], v(A^t_i) > 0\).
    \end{proof}
    Our result follows.
\end{proof}

Lastly, we consider the bi-valued goods setting. In this setting, we define two positive constants, and each agent values each good at one of them. Notice that we exclude zero-valued goods; otherwise, it degenerates into the generalized binary valuation setting. We propose that the approximation ratio is the ratio between the two constants. Surprisingly, this ratio is tight.

\subsubsection{Bi-valued goods}
\begin{theorem}
\label{thm:atefxbivalued}
For bi-valued goods valued \(b \geq a\), the Round-Robin algorithm produces \(\frac{a}{b}\)-TEFX allocation without scheduling, and it is tight.
\end{theorem}
We first prove its correctness, then prove its tightness.
\begin{proof}
Consider an agent \(i\), where RR allocates the next good to them.\\
\textbf{Case 1:} \(j > i\). Since \(\frac{b}{a} \cdot v_i(g_1) \geq v_i(g_2), \forall g_1 \in A_i\) and \(\forall g_2 \in A_j\), and by Observation~\ref{thm:RR3}, \(|A_i| = |A_j|\), there exists a perfect matching between each \(g_1 \in A_i\) and \(g_2 \in A_j\). Therefore, 
\[
\sum_{g_1 \in A_i} \frac{b}{a} \cdot v_i(g_1) \geq \sum_{g_2 \in A_j} v_i(g_2)
\]
where we have
\[
\frac{b}{a} \cdot v_i(A_i) \geq v_i(A_j) \;\Rightarrow\; v_i(A_i) \geq \frac{a}{b} \cdot v_i(A_j) \text{.}
\]
\textbf{Case 2:} \(j < i\). By Observation~\ref{thm:RR3}, \(|A_i| = |A_j| - 1\). Hence, by removing an arbitrary good from \(A_j\), we return to \textbf{Case 1}. Therefore, \(v_i(A_i) \geq \frac{a}{b} \cdot v_i(A_j \backslash \{g\}), \forall g \in A_j\).

Now, we prove its tightness using an example. Consider the instance with two agents and three goods, spanning two rounds, where the agents have identical valuations. \(O_1 = \{g_1, g_2\}\) and \(O_2 = \{g_3\}\). \(v(g_1) = v(g_2) = a\) and \(v(g_3) = b\).

After \(t = 1\), each agent must be allocated each of the goods, respectively, to obtain a positive approximation of TEFX. After \(t = 2\), without loss of generality, assume agent \(1\) obtains the good. \(v_1(A_1) = a + b \geq v_1(A_2) = a\). On the other hand, after removing any good from \(A_1\), the worst case is removing the good valued \(a\) for agent \(2\), where we have \(v_2(A_1) - a = b\). Since we have
\[
v_2(A_2) \geq \frac{a}{b} \cdot v_2(A_1) = \frac{a}{b} \cdot b = a \text{,}
\]
\(\frac{a}{b}\)-TEFX is tight.
\end{proof}

In summary, our exploration of $\alpha$-TEFX demonstrates that while exact temporal envy-freeness is overwhelmingly brittle in dynamic environments, multiplicative approximations offer a robust and mathematically viable alternative. We established that universal constant-factor guarantees are generally impossible due to inherent valuation variance and the specific challenges posed by zero-valued goods. However, by introducing structural constraints such as the periodicity of identical days or the rigid utility tiers of generalized binary valuations, we successfully identified scenarios in which strong constant-factor approximations, such as $1/2$-TEFX, can be recovered. Even in environments where constant bounds remain elusive, such as identical valuation settings, carefully designed greedy mechanisms allow us to mitigate extreme disparities and secure optimized, value-dependent guarantees. Ultimately, these results confirm that relaxing the fairness criteria from exact to multiplicative allows us to bridge the gap between theoretical impossibility and practical, algorithmically sound fair division.

\subsection {TMMS} \label{tmms}
In this section, we consider Temporal Maximin Share Fairness (TMMS). Note that MMS does not exist when there are more than two agents \citep{kurokawa2016,kurokawa2018}. We are the first to consider MMS under the temporal fair division model. 

Unfortunately, under most settings, a TMMS allocation may not exist. In Theorem~\ref{thm:tmmsgeneralsettingR}, we show that a TMMS allocation may not exist generally, even under restricted settings and scheduling.

Here, we note that under the generalized binary valuation setting, EF1 implies MMS \citep{bouveret2016}. Therefore, since \cite{elkind2024} proved the possibility for TEF1 under the generalized binary valuation setting, TMMS also exists under this setting.

\subsubsection{Identical Days}
We consider the identical days setting. Here, we show that a TMMS allocation may not exist.

\begin{theorem}
    \label{thm:tmmsidenticaldays}
    A TMMS allocation for goods under the identical days setting may not exist, even when there are two agents with identical valuations, and \(T > 1\).
\end{theorem}
\begin{proof}
    Consider the instance with two agents and three goods that arrive in each round \(\forall t \in [T], O_t = \{g^t_1, g^t_3, g^t_{10}\}\), where agents have identical valuations: \(v(g^t_1) = 1, v(g^t_3) = 3, v(g^t_{10}) = 10\).
    
    In round \(1\), in order to obtain a TMMS allocation, each agent must be allocated either \(\{g^1_1, g^1_3\}\) or \(\{g^1_{10}\}\).\\
    In round \(2\), each agent must be allocated either \(\{g^1_1, g^1_3, g^1_{10}\}\) or \(\{g^2_1, g^2_3, g^2_{10}\}\) to obtain a TMMS allocation.
    
    In order for the allocation at the end of round \(3\) to be TMMS, one agent must be allocated \(\{g^1_1, g^2_1, g^1_3, g^2_3, g^3_3, g^1_{10}\}\), and the other agent must be allocated \(\{g^3_1, g^2_{10}, g^3_{10}\}\), which is not possible, as \(g^1_3\) and \(g^2_3\) are allocated to different agaents in round \(2\).
\end{proof}

In summary, our investigation into Temporal Maximin Share Fairness (TMMS) reveals that the cumulative nature of maximin guarantees is exceptionally difficult to satisfy in a dynamic context. Unlike static settings, where MMS is often a more accessible benchmark than EFX, the requirement to meet the maximin threshold in every round makes TMMS remarkably fragile. We demonstrated that even with only two agents and identical valuations, the sequential arrival of goods, whether in general settings or repeating identical days, frequently precludes the existence of a TMMS allocation. The sole sanctuary of this fairness notion lies in the highly structured domain of generalized binary valuations, where the existence of TEF1 suffices for a positive result. Ultimately, these findings suggest that while TMMS is a theoretically interesting extension of the maximin principle, it is often too restrictive for general temporal allocation, reinforcing the need for the approximation-based or preference-restricted approaches we have explored.

\section{Temporal Fair Division with Scheduling} \label{wschedule}
Scheduling is a process that allows a good that arrives in round \(t\) to be scheduled in round \(t'\) such that \(t' \geq t\) and \(t' - t + 1 \leq r\). We define this process with a function \(delay(g, t, r')\), where \(r' \leq r\), which denotes \(O_t \gets O_t \backslash \{g\}\) and \(O_{t+r'-1} \gets O_{t+r'-1} \cup \{g\}\). Note that when \(r = 1\), it is equivalent to standard temporal fair division, and when \(r = T\), it is equivalent to standard fair division.

We compare an instance $\mathcal{I}_1$ (no scheduling, $T=T' \cdot k$ rounds, \(O_t = \{g_{(t-1) \cdot k + 1}, \dots, g_{t \cdot k}\}\)) with an instance $\mathcal{I}_2$ (scheduling allowed, \(T = T'\), $r \geq k$, \(O_t = \{g_t\}\)). Note that both instances have the same number of goods that arrive across all rounds.

First, under the identical days setting, $\mathcal{I}_1$ and $\mathcal{I}_2$ are incomparable unless $T \equiv 0 \pmod k$, as $\mathcal{I}_2$ cannot schedule identical bundles across rounds if the goods cannot be scheduled into another identical days setting instance.

For general settings, $\mathcal{I}_1 \implies \mathcal{I}_2$. We can simply schedule the goods from $\mathcal{I}_1$ into $k$ rounds of $\mathcal{I}_2$ and leave the remaining rounds empty. Since empty rounds preserve existing fairness, any solution for $\mathcal{I}_1$ is valid for $\mathcal{I}_2$.

However, $\mathcal{I}_2 \not\implies \mathcal{I}_1$. Scheduling allows goods to be separated to avoid ``fairness traps" present in fixed batches. Consider a counter-example for TEFX with two identical valuation agents ($v(g_1)=v(g_2)=1, v(g_3)=100, v(g_4)=10$).
\begin{itemize}
    \item In $\mathcal{I}_1$ (Fixed): If $O_1=\{g_1, g_2\}$ and $O_2=\{g_3, g_4\}$, \(T = 2\), agents must split $\{g_1, g_2\}$ in Round 1. In Round 2, receiving $g_3$ creates unavoidable envy that cannot be eliminated by removing a small good from Round 1, violating TEFX.
    \item In $\mathcal{I}_2$ (Scheduled): If \(\forall 1 \leq t \leq 4, O_t = \{g_t\}\), \(r = 2\), \(T = 4\), and we schedule $g_2$ to Round 3, we can allocate $A_1=\{g_1, g_2, g_4\}$ and $A_2=\{g_3\}$. This sequence maintains TEFX at every step, a solution impossible in the instance without scheduling.
\end{itemize}
\subsection{TEF1} \label{tef1R}
We observe that if we are able to schedule the goods under the identical days setting, a TEF1 allocation exists.

\subsubsection{Identical Days}
\begin{theorem}
    \label{thm:tef1identicalDaysRn/2}
    Algorithm~\ref{alg:tef1identicalDaysRn/2} returns a TEF1 allocation in polynomial time under the identical days setting with scheduling when \(r \geq \frac{n}{2}\). In addition, every \(n\) rounds, it returns a TEF allocation.
\end{theorem}
\begin{algorithm}
    \caption{RR'}
    \label{alg:tef1identicaldaysRn/2RR}
    \textbf{Input}: Arbitrary fair division instance \(\mathcal{I} = \left( N, M, v = (v_1, \ldots, v_n)\right)\)\\
    \textbf{Output}: EF1 allocation \(\mathcal{A}\)
    \begin{algorithmic}[1]
        \For{$i \in N$}
            \State $A_i \gets \emptyset$
        \EndFor
        \For{$l = 1$ to $m$}
            \State $i \gets (i$ mod $m) + 1$
            \State $g \gets$ The most valued good for agent \(i\), and its identical copy is not in \(A_i\)
            \State $A_i \gets A_i \cup \{g\}$
            \State $M \gets M \backslash \{g\}$
        \EndFor
    \end{algorithmic}
\end{algorithm}
\begin{proof}
    The polynomial runtime of Algorithm~\ref{alg:tef1identicalDaysRn/2} is easy to verify, given that there are only nested \textbf{For} loops, which run in polynomial time, Algorithm~\ref{alg:tef1identicaldaysRn/2RR} runs in polynomial time, and the other operations run in polynomial time. Thus, we focus on proving correctness.
    
    Denote \(\lambda = n\). 
    \begin{algorithm}
    \caption{Returns a TEF1 allocation when \(r \geq \frac{n}{2}\)}
    \label{alg:tef1identicalDaysRn/2}
    \textbf{Input}: Arbitrary temporal fair division instance \(\mathcal{I} = \left( N, T, \left\{ O_t \right\}_{t \in [T]}, v = (v_1, \ldots, v_n), r \right)\)\\
    \textbf{Output}: TEF1 allocation \(\mathcal{A}\)
\begin{algorithmic}[1]
    \State $k \gets 0$
    \State $cpy \gets \emptyset$
    \For{$g \in O_1$}
        \State Append $g$ to $cpy$
    \EndFor
    \State $len \gets |cpy|$
    \While{$k \cdot n + n - 1 < T$}
        \For{$t = 0$ to $\lceil\frac{n}{2}\rceil - 1$}
            \For{$g \in O_{t + k}$}
                \State $delay(g, t + k, (k \cdot n + \lceil\frac{n}{2}\rceil) - (t + k) + 1)$
            \EndFor
        \EndFor

         \State \(A^{k \cdot n + \lceil\frac{n}{2}\rceil} \gets\) RR'\((N, O_{k \cdot n + \lceil\frac{n}{2}\rceil}, v)\)
        
        \For{$t = \lceil\frac{n}{2}\rceil + 1$ to $n - 1$}
            \For{$g \in O_{t + k}$}
                \State $delay(g, t + k, (k \cdot n + n - 1) - (t + k) + 1)$
            \EndFor
        \EndFor

        \State Allocate \(O_{k \cdot n + n - 1}\) such that
        \Statex \(\forall i, j \in N, v_i(A^{k \cdot n + n - 1}_i) = v_i(A^{k \cdot n + n - 1}_j)\)
        \State $k \gets k + 1$
    \EndWhile
    \For{$t = k$ to $\frac{T + k}{2} - 1$}
        \For{$g \in O_t$}
            \State $delay(g, t, \frac{T + k}{2} - t + 1)$
        \EndFor
    \EndFor
    \State $A^{\frac{T + k}{2}} \gets$ RR$(N, O_{\frac{T + k}{2}}, v)$
    \For{$t = \frac{T + k}{2} + 1$ to $T - 1$}
        \For{$g \in O_t$}
            \State $delay(g, t, T - t + 1)$
        \EndFor
    \EndFor
    \State Reverse $N$
    \State $A^T \gets$ RR$(N, O_T, v)$
\end{algorithmic}
\end{algorithm}
    Intuitively, Algorithm~\ref{alg:tef1identicalDaysRn/2} schedules all of the goods such that every \(\lambda\) rounds, it is scheduled into two phases, where in each phase it contains approximately \(\frac{\lambda}{2}\) rounds' cumulative goods. We attempt to maintain a state where, after every two phases (\(\lambda\) rounds), each agent has exactly one copy of each good. To achieve this, for the first phase, we need to obtain an EF1 allocation, ensuring that no agent receives more than one copy of each good. At last, if \(T\) is not divisible by \(\lambda\), we schedule the goods into two rounds, applying the result from \cite{elkind2024} to compute a TEF1 allocation. We note that the result from \cite{elkind2024} cannot be applied to allocating scheduled phases, as our result relies on the fact that in every \(\lambda\) rounds, we obtain a TEF allocation, whereas their result does not guarantee this.
    
    We show a key property in Algorithm~\ref{alg:tef1identicaldaysRn/2RR}.
    \begin{lemma}
        \label{lemma:tef1identicalDaysRn/21}
        Algorithm~\ref{alg:tef1identicaldaysRn/2RR} returns an EF1 allocation if there are \(\frac{n}{2}\) copies of each good.
    \end{lemma}
    \begin{proof}
        First, we show that there always exists a good to be allocated to an agent. Denote \(\gamma\) to be the number of goods that arrive in each round. Now, we define a ``RR-round" as each round in RR. Let \(\beta\) be the current RR-round. Notice the total number of goods is \(\gamma \cdot \frac{\lambda}{2}\), which means there are \(\frac{\gamma}{2}\) RR-rounds in total. Hence, \(\beta \leq \frac{\gamma}{2}\). Then, for an agent, there are \(\gamma - \beta\) types of goods left that can be allocated to them after RR-round \(\beta\). Suppose to contradictory that for this agent, they cannot be allocated any good. This means all other \(n - 1\) agents have taken \((\lambda - 1) \cdot (\gamma - \beta)\) goods. Notice that \(\gamma - \beta \geq \frac{\gamma}{2}\) since \(\beta \leq \frac{\gamma}{2}\). This means at least \(\frac{\gamma}{2}\) RR-rounds have passed, and on the next RR-round, this agent cannot be allocated any good. However, since there are only at most \(\frac{\gamma}{2}\) RR-rounds, there is no such next RR-round, thus we are always able to allocate to an agent.
        
        Now, the EF1 property in Algorithm~\ref{alg:tef1identicaldaysRn/2RR} is trivial if we are able to guarantee that each agent can always be allocated a good.
    \end{proof}
    
    Next, we show a key property in Algorithm~\ref{alg:tef1identicalDaysRn/2}.
    \begin{lemma}
        \label{lemma:tef1identicalDaysRn/2}
        Every \(n\) rounds, Algorithm~\ref{alg:tef1identicalDaysRn/2} allocates each agent the goods that arrive in each round. Formally, denote \(O^* = O_1\), which is equivalent to all other rounds as we are under the identical days setting. Every \(n\) rounds, each agent is allocated \(O^*\). This means \(A^t_i = O^* \cup A^{t - n}_i\).
    \end{lemma}
    \begin{proof}
         Consider the goods at the \(k \cdot n + \lceil\frac{n}{2}\rceil, k \in \mathbb{N}\) round to be \(\lceil\frac{n}{2}\rceil\) copies of the same good. By Lemma~\ref{lemma:tef1identicalDaysRn/21}, Algorithm~\ref{alg:tef1identicaldaysRn/2RR} returns an EF1 allocation without allocating duplicate copies for each agent. Therefore, we do not allocate multiple copies to the same agent. Every \(n\) rounds, for each agent, we simply allocate a good to an agent if they do not have this ``copy", and each agent has exactly one copy of this good. After this allocation, each agent is allocated exactly the goods that arrive in each round, and Lemma~\ref{lemma:tef1identicalDaysRn/2} follows.
    \end{proof}
    Then, we prove Theorem~\ref{thm:tef1identicalDaysRn/2} by induction. For the base case \(t = \lambda\), we split its proof into two parts. First, in round \(t = \lceil\frac{\lambda}{2}\rceil\), we obtain an EF1 allocation. Then, in round \(t = \lambda\), by the algorithm's design, we allocate goods for each agent that are equivalent to the goods that arrive in each round. It is clear that at \(t = \lambda\), we have an EF allocation, as each agent has the same-valued allocation.
    
    Then, we prove the induction step. Assume that for an arbitrary \(t = k \cdot \lambda\), for any agent \(i\), \(\forall j \in [N], v_i(A^{k \cdot \lambda}_i) = v_i(A^{k \cdot \lambda}_j)\). In round \(t = k \cdot \lambda + \lceil\frac{\lambda}{2}\rceil\), since the Round-Robin algorithm returns an EF1 allocation, for any arbitrary agent \(i\), \(\forall j\), \(\exists g \in A^t_j \backslash A^{k \cdot n}_j\) s.t. 
    $v_i(A^t_i \backslash A^{k \cdot \lambda}_i) \geq v_i(A^t_j \backslash A^{k \cdot \lambda}_j \backslash \{g\})$ \text{.}
    By adding the excluded part from both sides of the inequality, we have 
    \[
    v_i(A^t_i \backslash A^{k \cdot \lambda}_i) + v_i(A^{k \cdot \lambda}_i) \geq v_i(A^t_j \backslash A^{k \cdot \lambda}_j \backslash \{g\}) + v_i(A^{k \cdot \lambda}_i) \text{.}
    \]
    Since \(v_i(A^{k \cdot \lambda}_i) = v_i(A^{k \cdot n}_j)\), we have \(v_i(A^t_i)\geq v_i(A^t_j \backslash \{g\}))\).
    Hence, agent \(i\)'s envy towards agent \(j\) can be eliminated by removing a good \(g\) from agent \(j\)'s allocation: \(v_i(A^t_i \backslash A^{k \cdot \lambda}_i \cup A^{k \cdot \lambda}_i) = v_i(A^t_i) \geq v_i((A^t_j \backslash A^{k \cdot \lambda}_j \cup A^{k \cdot \lambda}_j) \backslash \{g\}) = v_i(A^t_j \backslash \{g\})\), which results in a TEF1 allocation.
    
    Now, as above, in round \(t = (k + 1) \cdot \lambda\), Algorithm~\ref{alg:tef1identicalDaysRn/2} allocates the same valued bundle for each agent, which results in a TEF allocation.\\
    Lastly, if \(T\) is not divisible by \(\lambda\), we can schedule the remaining rounds as two rounds as \(r \geq \frac{\lambda}{2} \geq \frac{T}{2}\). By \cite{elkind2024}, by running an RR in the first scheduled round, and another RR with the allocation order reversed in the second (last) scheduled round, we obtain a TEF1 allocation.
\end{proof}

\subsection{TEFX} \label{tefxR}
 Since a TEFX allocation may not exist without scheduling, the next question is whether TEFX exists when scheduling is allowed. We answer this question by providing a negative answer, provided the instance does not degenerate into the standard EFX model, where there exists only one round.

\begin{theorem}
\label{thm:tefxgeneralsettingsR}
    A TEFX allocation may not exist, even with scheduling, when there are only two agents with identical valuations, for \(r \leq T-1\) and \(T > 1\).
\end{theorem}
\begin{proof}
Consider the instance with two agents and three goods \(O = \{g_1, g_2, g_3\}\), where agents have identical valuations: \(v(g_1) = v(g_2) = 1\) and \(v(g_3) = 2\), and \(r \leq T - 1\). \(O_1 = \{g_1, g_2\}\), \(\forall 1 < t < T, O_t = \{\emptyset\}\), and \(O_T = \{g_3\}\). Essentially, we insert rounds in which no goods arrive, such that no matter how the goods are scheduled, \(g_1\) and \(g_2\) must be allocated before \(g_3\).

To prove that TEFX does not exist in this instance, we encourage interested readers to read the proof in \cite{elkind2024}, as our proof closely resembles it. Specifically, when \(T = 2\), \(r \leq 1\), which is the case where there are no additional rounds, it is equivalent to the case in their proof.
\end{proof}

\subsection{TMMS} \label{tmmsR}
We consider TMMS with scheduling, and show that even when there are two agents with identical valuations, after scheduling, TMMS may not exist, provided the instance does not degenerate into a standard MMS setting.

\begin{theorem}
    \label{thm:tmmsgeneralsettingR}
    A TMMS allocation for goods may not exist, even with scheduling, when there are only two agents with identical valuation, for \(r \leq T - 1\), and \(T > 1\).
\end{theorem}
\begin{proof}
    We use the same instance from Theorem~\ref{thm:tefxgeneralsettingsR}. For completeness, we state the instance here again: Consider the instance with two agents and three goods \(O = \{g_1, g_2, g_3\}\), where agents have identical valuations: \(v(g_1) = v(g_2) = 1\) and \(v(g_3) = 2\), and \(r \leq T - 1\). \(O_1 = \{g_1, g_2\}\), \(\forall 1 < t < T, O_t = \emptyset\), and \(O_T = \{g_3\}\). As mentioned in the proof of Theorem~\ref{thm:tefxgeneralsettingsR}, we insert rounds in which no goods arrive, such that no matter how the goods are scheduled, \(g_1\) and \(g_2\) must be allocated before \(g_3\).
    
    In round \(1\), to obtain a TMMS allocation, each agent must be allocated either \(g_1\) or \(g_2\). Without loss of generality, assume agent \(1\) is allocated \(g_3\). However, notice that in order to obtain a TMMS allocation, each agent must be allocated either \(\{g_1, g_2\}\) or \(g_3\), which is impossible to achieve.
\end{proof}

\section {Discussion}
In this work, we explored the boundaries of fairness in temporal resource allocation, analyzing both approximation ratios and the power of scheduling. 
We show that while standard temporal fairness (such as TEFX and TMMS) remains largely unachievable even under restricted domains, meaningful guarantees are achievable if requirements are relaxed. Specifically, we established that scheduling is a powerful but limited tool. By introducing a buffer of size $r \geq n/2$, we can guarantee TEF1 in the identical days setting, providing a practical mechanism for system designers who can tolerate some amount of latency. However, we also showed that scheduling is not a panacea: the impossibility of TEFX and TMMS persists across most settings, suggesting that these metrics may be too demanding for temporal fair division. Regarding approximations without scheduling, our results for $\alpha$-TEFX reveal interesting properties. While constant-factor approximations are impossible in the general case, we proved that, under domain restrictions or in other settings such as generalized binary valuations, they admit stronger guarantees.

Several open questions remain regarding the extension of the temporal fair division framework. One key question is to tighten the scheduling bounds. Our result establishes $r \geq T/2$ as a buffer for TEF1, which may not be a tight constraint. It remains an open question whether a smaller scheduling buffer can be obtained.

\clearpage
\bibliographystyle{apalike}
\bibliography{bib}

\clearpage
\appendix
\section{Appendix}

\subsection{Pseudocode of ECE}
\begin{algorithm}
    \caption{ECE}
    \label{alg:ECE}
    \textbf{Input}: Arbitrary fair division instance \(\mathcal{I} = \left( N, M, v = (v_1, \ldots, v_n)\right)\)\\
    \textbf{Output}: Allocation \(\mathcal{A}\)
\begin{algorithmic}[1]
    \For{$i \in N$}
        \State $A_i \gets \emptyset$
    \EndFor
    \For{$l = 1$ to $m$}
        \While{there does not exist an unenvied agent}
            \State Find an envy-cycle \(C = \left(i_0, \dots, i_{d-1}\right)\)
            \State \(\forall 0 \leq j < d\), assign the allocation of \(i_{j+1 \text{mod} d}\) to \(i_j\).
        \EndWhile
        \State Let $i$ be an unenvied agent
        \State Let $g^*$ be an arbitrary good
        \State $A_i \gets A_i \cup \{g^*\}$
        \State $M \gets M \backslash \{g^*\}$
    \EndFor
\end{algorithmic}
\end{algorithm}

\subsection{Psuedocode of MAX-ECE}
\begin{algorithm}
    \caption{MAX-ECE}
    \label{alg:MAX-ECE}
    \textbf{Input}: Arbitrary fair division instance \(\mathcal{I} = \left( N, M, v = (v_1, \ldots, v_n)\right)\)\\
    \textbf{Output}: Allocation \(\mathcal{A}\)
\begin{algorithmic}[1]
    \For{$i \in N$}
        \State $A_i \gets \emptyset$
    \EndFor
    \For{$l = 1$ to $m$}
        \While{there does not exist an unenvied agent}
            \State Find an envy-cycle \(C = \left(i_0, \dots, i_{d-1}\right)\)
            \State \(\forall 0 \leq j < d\), assign the allocation of \(i_{j+1 \text{mod} d}\) to \(i_j\).
        \EndWhile
        \State Let $i$ be an unenvied agent
        \State Let $g^* = arg max_{g \in M} v_i(g)$
        \State $A_i \gets A_i \cup \{g^*\}$
        \State $M \gets M \backslash \{g^*\}$
    \EndFor
\end{algorithmic}
\end{algorithm}

\subsection{Psuedocode of RR}
\begin{algorithm}
    \caption{RR}
    \label{alg:RR}
    \textbf{Input}: Arbitrary fair division instance \(\mathcal{I} = \left( N, M, v = (v_1, \ldots, v_n)\right)\)\\
    \textbf{Output}: Allocation \(\mathcal{A}\)
    \begin{algorithmic}[1]
        \For{$i \in N$}
            \State $A_i \gets \emptyset$
        \EndFor
        \For{$l = 1$ to $m$}
            \State Let $i \gets N_{l \text{mod} n}$
            \State Let $g^* \in arg max_{g \in M} v_i(g)$
            \State $A_i \gets A_i \cup \{g^*\}$
            \State $M \gets M \backslash \{g^*\}$
        \EndFor
    \end{algorithmic}
\end{algorithm}

\end{document}